\documentclass[journal,final,twocolumn,10pt,twoside]{IEEEtran}

\IEEEoverridecommandlockouts
\usepackage[noadjust]{cite}
\usepackage{amsmath,amssymb,amsfonts,amsbsy,amsthm}
\usepackage{algorithmic}    
\usepackage{graphicx}
\usepackage{textcomp}
\usepackage{xcolor}
\usepackage{siunitx}
\usepackage{balance}
\usepackage{lipsum}
\usepackage[caption=false,font=footnotesize]{subfig} 
\usepackage{algorithm}

\usepackage[footnotesize]{caption}

\newtheorem{theorem}{Theorem}
\newtheorem{lemma}{Lemma}

\newtheorem{definition}{Definition}

\newtheorem{corollary}{Corollary}

\definecolor{mblue}{rgb}{0,0.1,1.0}
\definecolor{usc}{rgb}{0.7, 0, 0}    
\definecolor{mred}{rgb}{0.9,0.1,0.1}
\definecolor{mpurple}{rgb}{0.6,0.1,0.8}
\definecolor{morange}{rgb}{0.8,0.4,0.1}
\definecolor{mblack}{rgb}{0,0,0}

\newcommand{\mcg}[1]{{\color{mblack}{#1}}}

\DeclareMathOperator*{\argmax}{arg\,max}
\DeclareMathOperator*{\argmin}{arg\,min}
    
\begin{document}

\title{\huge Two-sided Delay Constrained Scheduling:\\ Managing Fresh and Stale Data
\\
\thanks{Portions of this paper have been \mcg{accepted} to the IEEE International Conference on Communications (ICC 2023) \cite{our_ICC_scheduling}. This work is funded in part by one or more of the following grants: 
NSF CCF-1817200, 
ARO W911NF1910269, 
DOE DE-SC0021417, 
Swedish Research Council 2018-04359, 
NSF CCF-2008927, 
NSF CCF-2200221, 
ONR 503400-78050, 
ONR N00014-15-1-2550, 
and the USC + Amazon Center on Secure and Trusted Machine Learning.}
}

\author{\large{Mustafa Can Gursoy and Urbashi Mitra}
	\\ \normalsize Department of Electrical and Computer Engineering, University of Southern California, Los Angeles, CA, USA
	\\ E-mails: \{mgursoy, ubli\}@usc.edu
}

\maketitle
\pagestyle{plain}

\begin{abstract}
Energy or time-efficient scheduling is of particular interest in wireless communications, with applications in sensor network design, cellular communications, and more. In many cases, wireless packets to be transmitted have deadlines that upper bound the times before their transmissions, to avoid staleness of transmitted data. In this paper, motivated by emerging applications in security-critical communications, age of information, and molecular communications, we expand the wireless packet scheduling framework to scenarios which involve strict limits on the time \emph{after} transmission, in addition to the conventional pre-transmission delay constraints. As a result, we introduce the scheduling problem under two-sided individual deadlines, which captures systems wherein transmitting too late (stale) and too early (fresh) are both undesired. Subject to said two-sided deadlines, we provably solve the optimal (energy-minimizing) offline packet scheduling problem. Leveraging this result and the inherent duality between rate and energy, we propose and solve the completion-time-optimal offline packet scheduling problem under the introduced two-sided framework. Overall, the developed theoretical framework can be utilized in applications wherein packets have finite lifetimes both before and after their transmission (\emph{e.g.,} security-critical applications), or applications with joint strict constraints on packet delay and information freshness.
\end{abstract}

\section{Introduction}
\label{sec:introduction}

\par Energy-efficient scheduling is a classical problem in wireless communications, which seeks to characterize and exploit the fundamental trade-off between transmission rate and packet delay \cite{schedule1,schedule2,schedule3}. In particular, for scheduling wireless communication packets over a single hop link, \cite{uysal2002_singledelay} addresses optimal offline scheduling with a single common transmission deadline, before which all packets should be transmitted. Expanding the single-deadline system, \cite{chen2008_predelay,individual_2,calculus_scheduling,individual_4,predelay_burst,individual_finite} consider each packet having its own individual delay constraint (deadline), before which it needs be processed. Energy-efficient scheduling under similar deadlines is also extended to multi-hop systems in \cite{wanshi_multihop}. As each delay constraint relates to, and limits, the time interval before the end of a packet's transmission \cite{chen2008_predelay,fangwen_predelay}, we refer to this conventional delay constraint as a \emph{pre-transmission delay constraint}, which ensures that the data is not \emph{stale}.

\par In this paper, we introduce an additional, alternative delay constraint that limits the time \emph{after} a packet has been processed by the scheduler. Vis-à-vis its conventional pre-transmission counterpart, we call this new deadline the \emph{post-transmission} delay constraint.  Modern applications, which will be described in the sequel,  necessitate ensuring that the data is not only not too stale, but also \emph{not too fresh.}
In addition to bounding delay for time-sensitive applications (\emph{e.g.,} video streaming), a conventional (pre-transmission) delay constraint can also be interpreted as ``expiration" at the scheduler node after which the packet becomes unusable. Thus, the scheduler needs to transmit it before its expiry. Using this interpretation, the post-transmission delay constraint arises from expiry at the \emph{destination} of the packet. In other words, a post-transmission delay constraint arises in scenarios wherein processing a packet ``too early" is undesirable, as it causes the packet to expire/become unusable at the next node that utilizes it. Combining the conventional pre-transmission deadlines with their post-transmission counterparts, we arrive at scheduling problems that are subject to \emph{two-sided delay constraints.} 
Several scenarios which incur two-sided delay constraints, wherein transmitting too late and too early are both undesirable include: \\
\noindent \textbf{Security-critical relaying:} Suppose a sender (S) is to send information towards a destination (D) that requires restricted access (\emph{e.g.,} a hash, private key, \emph{etc.}), through a relay-aided communication link. To ensure security, S transmits the information by fragmenting it into $M$ packets due reassembly at D. The goal of the relay (R, the scheduler) is to transmit each of the $M$ packets so that they are available at D at a certain desired time $t_R$ for successful reassembly. In case any node is compromised, the security protocol requires each packet to have a limited time it can spend at any node, after which it expires and its content is erased. Thus, R needs to transmit each packet before its individual deadline to avoid expiration (\emph{i.e.,} a pre-transmission delay constraint). As a similar finite lifetime is also employed at D, each packet needs to depart R \emph{later} than a certain deadline to avoid early expiry and ensure availability at $t_R$, yielding a post-transmission delay constraint on each packet. 

\noindent \textbf{Information Freshness:} A central controlling unit (the scheduler) seeks to achieve coordinated action among $M$ agents at a certain time $t_R$, via sending each agent a packet of instructions and status information regarding the system of interest. The packets need to be transmitted sufficiently in advance to allow for preparation prior to action, or due to packet delay limitations (pre-transmission delay constraints). Furthermore, the need for ensuring the agents to have fresh/up-to-date information regarding the system at $t_R$ requires avoiding transmitting ``too early", which corresponds to a post-transmission delay constraint on each packet.

\par \mcg{Note that given the bits of the packet can be changed by the scheduler until they are transmitted, and assuming the time-critical information is located at the end of the packet with a length that is negligible to the whole packet's length, the age of the time-critical information would be approximately zero at the time of departure. In such a system, a strict post-transmission deadline on the departure time would then upper bound the age of the time-critical information \cite{aoi} within the packet at time $t_R$.}

\noindent \textbf{Molecular communication:} In a physical layer molecular communication (MC) setup, a transmitter receives stimulus/input molecules in a time sequence \cite{stochastic_molecule_arrivals,queueing_in_MC,old_markovian}, converts them through a chemical reaction, and releases them as its output signal (akin to processing a job, with a molecule being the fundamental job unit and the transmitter being the scheduler). Input molecules exhibit degradation thus have finite lifetimes \cite{bacteria_OV_SG,degradation_1,degradation_2}, thus the MC transmitter is subject to pre-transmission delay constraints. The output molecules are also subject to molecular degradation, thus processing too early may cause the molecule to degrade before providing utility, corresponding to a post-transmission delay constraint. 

\par Herein, we first address the energy-efficient scheduling problem under this new two-sided constrained framework. In addition, we shall also consider the ``dual" problem, completion time minimization, which seeks to minimize the total time of the required packet transmissions subject to a fixed energy budget \cite{converse_original}. The version of this problem without individual delay constraints is addressed in \cite{converse_original}, wherein iterative and dynamic programming-based approaches are developed and empirically demonstrated to converge to the optimal offline allocation. The solutions of \cite{converse_original} are generalized in \cite{converse_energy_harvesting} to a scheduler with energy harvesting capabilities (\emph{i.e.,} the energy budget is refueled over time), and to packets with varying sizes. In \cite{comptime_wpcn,comptime_wpcn_coleri}, completion time minimization with energy harvesting is considered in the context of wireless powered communication networks, wherein an increased waiting time increases harvested energy and allows for a higher rate of data transmission. Thus, a trade-off between harvested energy and communication delay is considered and exploited for system optimization. In \cite{comptime_wpcn,individual_throughputmax}, a similar problem to the transmission completion time problem is formulated, which seeks to optimize the average throughput within a fixed time interval. In addition, \cite{individual_throughputmax} considers individual transmission deadlines, similar to pre-transmission delay constraints, in its formulation. None of these works consider strict deadlines that limit the time intervals \emph{after} transmission (\emph{i.e.,} ``not too fresh"), nor two-sided bounds on allowable departure times. Overall, to the best of our knowledge, energy- or time-efficient scheduling have not been previously considered under the two-sided delay constrained framework, which we treat herein. 

\par The contributions of this paper can be summarized as follows:
\begin{enumerate}
    \item We expand the current energy-efficient packet scheduling framework by incorporating the newly considered post-transmission delay constraints, and formulate the energy-efficient scheduling problem under two-sided delay constraints. \mcg{Though the developed framework is presented through a packet scheduling scenario, its theoretical structure is abstracted away from a specific application, and can be applied or adapted to whenever a set of jobs arrive as a sequence in time and exhibit two-sided deadlines that limit their completion time.}
    

    \item We discuss the feasibility of the scheduling problem, and provide a provably optimal offline scheduling algorithm that minimizes energy consumption under two-sided individual deadlines.
    \item We formulate and address the dual problem, in which the goal is to minimize the total transmission completion time subject to a fixed energy budget and the two-sided delay constraints. 
    \item We provide an algorithm that achieves provably optimal scheduling under the newly considered post-transmission delay constraints only.
    \item Following this, we cover pre-transmission delay constraints as well, and provide an offline algorithm that provably minimizes completion time under two-sided individual delay constraints. 
\end{enumerate}

\par 
In our preliminary work \cite{our_ICC_scheduling}, we formulated the energy-efficient scheduling problem under two-sided delay constraints, discussed its feasibility and provided the optimal offline scheduling algorithm. This paper extends and completes \cite{our_ICC_scheduling} by formulating the dual, time-minimization problem under two-sided delay constraints, providing the optimal offline scheduling algorithm and its accompanying optimality proof, and providing the full versions of the proofs in \cite{our_ICC_scheduling}.

\par The rest of the paper is organized as follows: Section \ref{sec:system_model} presents the two-sided delay constrained problem formulation and provides the definitions used throughout the paper. Section \ref{sec:prelims} recites prior literature on scheduling with no individual delays \cite{uysal2002_singledelay} in order to motivate the developed scheduling strategy. Section \ref{sec:two_sided_delay} presents our proposed offline scheduling algorithm under two-sided delay constraints, whose optimality is proved in Appendix \ref{ap:proof_of_twosided}. Section \ref{sec:completiontime} addresses the dual problem, where the scheduler seeks to minimize the total completion time subject to an energy cost. The section provides algorithms subject to post-delays only (whose optimality proof is in Appendix \ref{ap:proof_of_comptime_post}), as well as to two-sided delays (optimality proof in-line). Section \ref{sec:numericalresults} presents numerical results on energy and time minimization under two-sided deadlines, and Section \ref{sec:conclusions} concludes the paper.

\section{System Model}
\label{sec:system_model}

\par Herein, we consider a scheduler that inputs packets in a certain temporal sequence\footnote{We assume packets of equal sizes in this paper.} with arrival times denoted by $t_i$, $i=1,\dots,M$, where $M$ is the total number of arriving packets. Without loss of generality, we assume $t_1 \doteq 0$ throughout the paper. Arriving packets are processed and transmitted by the scheduler, or if the scheduler is busy serving another packet, wait in queue until transmission. Throughout the paper, \mcg{similar to our key prior works \cite{uysal2002_singledelay,chen2008_predelay,converse_original} we consider first-in, first-out (FIFO) processing, which implies that the packets are processed in the order they are received}\footnote{We leave scheduling under arbitrary, non-FIFO ordering for future work.}. To illustrate the setup, Figure \ref{fig:systemdiag} is presented to provide an instantiation of packet arrival and departure times.

\begin{figure}[!t]
    \centering
    \includegraphics[width = 0.43\textwidth]{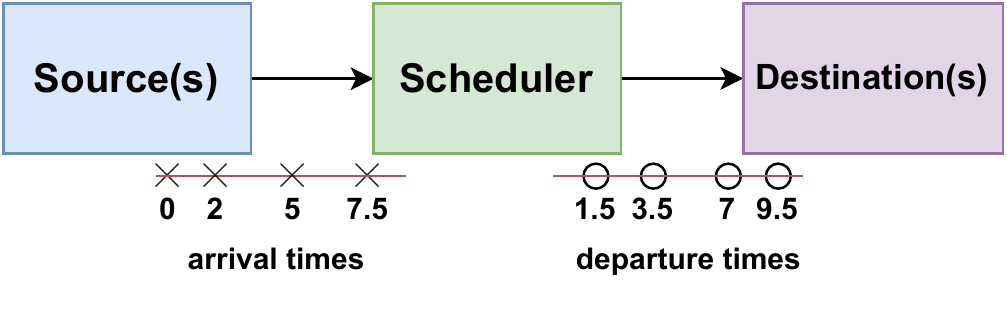}
    \caption{\mcg{The system model of interest with $M=4$. Packet arrival times at scheduler $\boldsymbol{t} = [0,2,5,7.5]$ due to some source(s), and end of transmission (departure) times $[1.5,3.5,7,9.5]$, possibly towards one or many destination(s).}}
    \label{fig:systemdiag}
\end{figure}

\par The vector $\boldsymbol{\tau} = [\tau_1 , \dots , \tau_M]^\top$ is defined as the vector that holds the processing times of packets. That is, for each packet $i$, $\tau_i$ denotes the time it takes for the scheduler to transmit the packet. For FIFO ordering, the departure time of packet $i$ is then equal to $\sum_{j=1}^{i} \tau_j$.

\par In the model we consider herein, it is assumed that each packet has a finite lifetime that can spend on i) the scheduler, and ii) the destination after it is transmitted. We assume that after this lifetime has expired, the content of the packet is erased from the node and the packet becomes unusable due to expiry. The goal of the scheduler is then to:
\begin{itemize}
    \item transmit every packet by a specified \emph{end time} $t_E$,
    \item optimize a certain cost associated with transmission durations $w(\boldsymbol{\tau})$ (\textit{e.g.,} energy consumption of transmitting the packet),
    \item and satisfy each packet's expiry-induced deadlines at both the scheduler and destination.
\end{itemize}
\textbf{The Cost:} Throughout the paper, \mcg{similar to the considerations of \cite{uysal2002_singledelay,chen2008_predelay,converse_original,converse_energy_harvesting,individual_throughputmax,equalcost_ref}}, we assume that the total cost incurred by all $M$ packet transmissions, $w(\boldsymbol{\tau})$, to have the following properties:
\begin{enumerate}
    \item The overall cost incurred by the system is the sum of individual costs of each transmission $w(\boldsymbol{\tau}) = \sum_{i=1}^{M} w(\tau_i)$.
    \item For a particular argument $\tau$, $w(\tau)$ is strictly convex and decreasing in $\tau$.\footnote{For wireless applications, the convexity of $w(\tau)$ can be motivated by inverting the Shannon's capacity formula (see \cite[Equation (1.1)]{wanshi_thesis} and the discussion therein for further details).}
    \item $w(\tau_i) > 0$ for all $i$.
    \item \mcg{Given $\tau_i$, $w(\tau_i)$ does not depend on time index $i$.}
\end{enumerate}
Bullet Point 2 implies that transmitting a packet faster incurs an increasingly higher cost. Throughout the manuscript, we will use the words cost and energy interchangeably. \\
\textbf{Delay Constraints:} We consider the communication link to be successful with respect to a chosen reference time $t_R$, if the following criteria are met at time $t_R$:
\begin{enumerate}
    \item Before any packet expires, the scheduler has successfully transmitted them to the destination,
    \item All transmitted packets have successfully arrived at the destination by the reference time $t_R$,
    \item At $t_R$, all packets are active (\textit{i.e.,} have survived and are not yet expired) at the destination.
\end{enumerate}
Given a packet's arrival time $t_i$, the first and second constraints together impose a hard deadline on the time at which packet $i$ departs the scheduler. In particular, to avoid either expiry at the scheduler or ensure that it successfully arrives to the destination by $t_R$, the scheduler has to upper bound the packet's allowed time for processing. As this phenomenon relates to the time interval \emph{before} transmission, we call this constraint a \emph{pre-transmission delay constraint}, and denote it by $T_{\mathrm{pre},i}$, also referred to as a \emph{pre-delay} constraint for brevity. 
As will be detailed in Section \ref{sec:two_sided_delay}, we assume $T_{\mathrm{pre},i} > 0$, as otherwise, the scheduling problem becomes ill-posed.
Furthermore, we consider the end time of the scheduling frame $t_E = t_M + T_{\mathrm{pre},M}$ throughout the paper.

\par In contrast to the first two conditions, the third constraint imposes a \emph{lower} bound on the departure time of packet $i$. Note that if a packet is transmitted too early, the transmitted packet would not survive until $t_R$ due to expiry. As this time interval limits the time interval \emph{after} a packet's transmission, we call this constraint a \emph{post-transmission delay constraint}, and denote it by $T_{\mathrm{post},i}$. Similar to its pre-delay counterpart, throughout the paper, we also refer to this constraints as the \emph{post-delay} constraint for brevity. Overall, Figure \ref{fig:twosideddelay_demo} depicts the feasible region of transmission imposed by these two constraints.
\begin{figure}[!t]
    \centering
    \includegraphics[width = 0.45\textwidth]{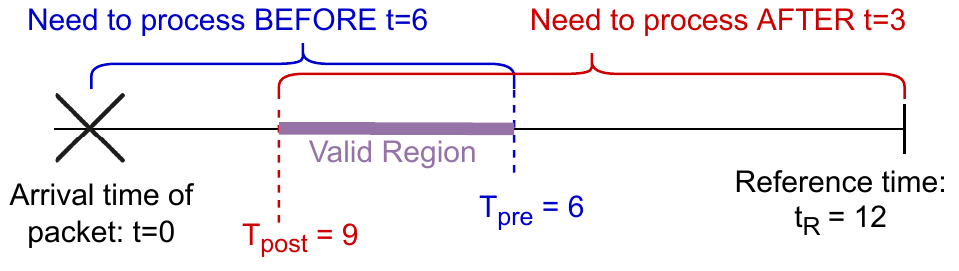}
    \caption{The valid departure time region induced by pre- and post-delay constraints. The packet arrives at time $t=0$, and the reference time is $t_R=12$. Pre- and post-delay constraints are $T_{\mathrm{pre}}=6$ and $T_{\mathrm{post}}=9$, respectively, resulting in a valid departure time region of $[3,6]$.}
    \label{fig:twosideddelay_demo}
\end{figure}

\par As a convention, we define $t_{M+1} \doteq t_E$, and define the time between two packet arrivals (\emph{i.e.,} the inter-arrival time) as $d_i = t_{i+1} - t_{i}$. Thus, for energy minimization, we consider the following optimization problem:
\begin{subequations}
	\begin{align}
		\min\limits_{\boldsymbol{\tau}} \quad & w(\boldsymbol{}{\tau}) = \sum\nolimits_{i=1}^{M} w(\tau_i)\\	
		\text{s.t.} \quad  & \sum\nolimits_{i=1}^{k} \tau_{i} \geq \sum\nolimits_{i=1}^{k} d_{i}, \quad k \in \{1, \dots, M-1\}, \label{eq:optimization_nonidling1} \\ 
		& \sum\nolimits_{i=1}^{M} \tau_{i} = t_E = \sum\nolimits_{i=1}^{M} d_{i}, \label{eq:optimization_nonidling2} \\
		& \sum\nolimits_{i=1}^{k} \tau_{i}-\sum\nolimits_{i=1}^{k-1} d_{i} \leq T_{\mathrm{pre},k}, 		 \hspace{0.1cm} k \in \{1, \ldots, M \} \label{eq:optimization_predelay} \\
		& \sum\nolimits_{i=1}^{k} \tau_{i} \geq t_R - T_{\mathrm{post},k}, \quad k \in \{1, \dots, M\}. \label{eq:optimization_postdelay}
	\end{align}
    \label{eq:optimization_bothdelay}
\end{subequations}

\noindent Here, Conditions \eqref{eq:optimization_predelay} and \eqref{eq:optimization_postdelay} refer to the aforementioned pre- and post-transmission delay constraints, respectively. Conditions \eqref{eq:optimization_nonidling1} and \eqref{eq:optimization_nonidling2} are referred to as the \emph{non-idling} constraints \cite{chen2008_predelay} and they follow from the cost function $w(\tau)$ being decreasing in $\tau$. Note that, if we have $\sum_{i=1}^{k} \tau_{i} < \sum_{i=1}^{k} d_{i}$ for some $i$, this implies that the time interval $[\sum_{i=1}^{k} (d_{i} - \tau_i),\sum_{i=1}^{k} d_{i}]$ is unused. Then, $w(\tau)$ can be further decreased by using this unused interval, which implies that $\boldsymbol{\tau}$ cannot be optimal. Thus, in order to be eligible for optimality, an allocation has to satisfy the non-idling constraints (\eqref{eq:optimization_nonidling1}-\eqref{eq:optimization_nonidling2}). 

\par \mcg{We note that from a numerical solver's perspective, Condition \eqref{eq:optimization_nonidling1} is redundant as long as causality in departure/arrival times is satisfied ($\sum_{i=1}^{k} \tau_i > t_k = \sum_{i=1}^{k-1} d_i$). That is, the solver's solution finds the optimal solution that complies with \eqref{eq:optimization_nonidling1} without explicitly mentioning \eqref{eq:optimization_nonidling1}. However, the non-idling constraint \eqref{eq:optimization_nonidling1} gives a stricter lower bound on the departure time than the causality constraint since $\sum_{i=1}^{k} d_i \geq \sum_{i=1}^{k-1} d_i$, thus we present \eqref{eq:optimization_nonidling1} instead for future reference in our algorithm optimality proofs.}


\par Throughout the paper, we consider optimal \emph{offline} scheduling, in which an idealized scheduler is assumed to have non-causal and perfect information on the arrival times $t_i$ for all $i$, as well as perfect information on the delay constraints $T_{\mathrm{pre},i}$ and $T_{\mathrm{post},i}$.
We note that solutions to our formulation lower bounds $w(\boldsymbol{\tau})$ that can be achieved by a practical scheduler that is imperfect in delay constraint/arrival information, or is an online scheduling algorithm.

\section{Optimal Scheduling with a Single Common Deadline}
\label{sec:prelims}

\par 

To provide insight, we first address the special case of the problem in \eqref{eq:optimization_bothdelay}, wherein all packets are only subject to a single common deadline (no individual pre- or post-delay constraints). The goal of the scheduler is then to minimize $w(\boldsymbol{\tau}) = \sum_{i=1}^{M} w(\tau_i)$ subject only to the non-idling constraints described by \eqref{eq:optimization_nonidling1}-\eqref{eq:optimization_nonidling2}. Due to the absence of pre-delays (\emph{i.e.,} $T_{\mathrm{pre},i} = t_R-t_i$ for all $i$), the end time (also referred to as the ``single delay" in this section) is $t_E = t_R$. 
The optimal offline scheduling under a single delay has been addressed by \cite{uysal2002_singledelay}, whose strategy to compute the optimal $\boldsymbol{\tau}$ is presented in Algorithm \ref{alg:base}.
\begin{algorithm}[!t]
	\fontsize{10}{10}\selectfont
	\begin{algorithmic}[1]
			\caption{Optimal offline scheduling under a single common deadline (\cite{uysal2002_singledelay}).}
			\label{alg:base}
			\renewcommand{\algorithmicrequire}{\textbf{Inputs:}}
			\renewcommand{\algorithmicensure}{\textbf{Output:}}
			\REQUIRE $M$, $\boldsymbol{d}$, $t_E=t_R$ \\
			
			\STATE For the first packet, $m_{1}=\max\limits_{k \in[1, \ldots, M]} \left\{\frac{1}{k} \sum_{i=1}^{k} d_{i} \right\}$ \\
			\STATE The maximizing index $k_{1}=\max \left\{k: \frac{1}{k} \sum_{i=1}^{k} d_{i}=m_{1} \right\}$ \\
			\STATE Set $\tau_{i}=m_{1}, \quad \text { for all } i \in \{1,\dots,k_1\}$ \\ 
			\STATE $j \leftarrow 1$ \\
			\vspace{0.2cm} For rest of the packets, until the end of batch, \textbf{do} \\ 
			\STATE $m_{j+1}= \max_{k \in\left[1, \ldots, M-k_{j}\right]} \{\frac{1}{k} \sum_{i=1}^{k} d_{i+k_{j}} \}$ \\
			\STATE $k_{j+1}=k_{j}+\max \{k: \frac{1}{k} \sum_{i=1}^{k} d_{i+k_{j}}=m_{j+1} \}$ \\
			\STATE $\tau_{i}=m_{j}, \quad \text { for all } i \in \{k_{j-1} + 1, \dots,  k_{j} \}$ \\
			\STATE $j \leftarrow j+1$ \\
			\textbf{return} $\boldsymbol{\tau}$
	\end{algorithmic}
\end{algorithm} 

\par As also mentioned in \cite{uysal2002_singledelay,chen2008_predelay}, an interesting observation is that the optimal $\boldsymbol{\tau}$ is not dependent on the specific $w(\tau)$ function, as long as the cost is strictly convex, decreasing, and non-negative. To illustrate the approach of Algorithm \ref{alg:base}, we present two toy problems in Figure \ref{fig:alg_base_toyproblem}.
\begin{figure}[!t]
    \centering
    \includegraphics[width = 0.46\textwidth]{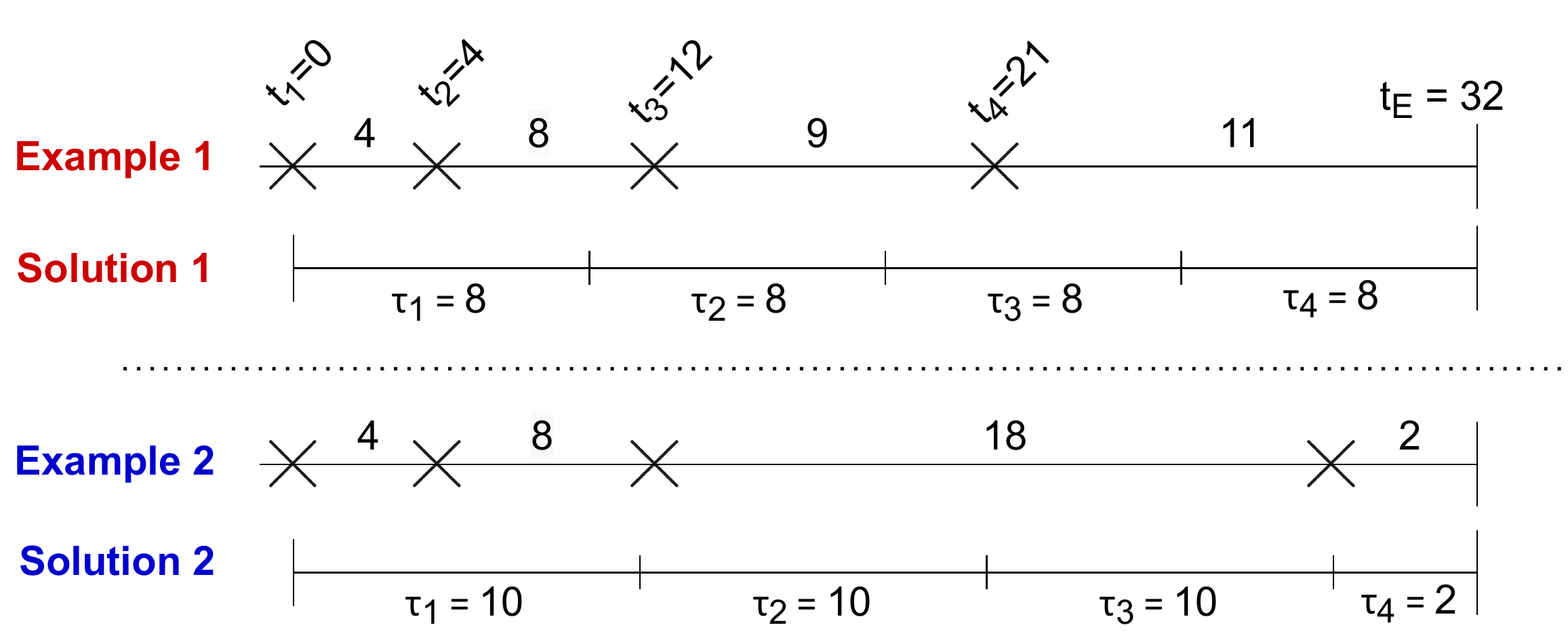}
    \caption{A demonstration of Algorithm \ref{alg:base} on two toy problems with $M=4$ and $t_R = t_E =  32$, under different $\boldsymbol{d}$ vectors.} 
    \label{fig:alg_base_toyproblem}
\end{figure}
\begin{itemize}
    \item \textit{Example 1:} When feasible in terms of the non-idling constraints, fully balancing each $\tau_i = 8$ minimizes a strictly convex, non-increasing $w(\boldsymbol{\tau})$. 
    \item \textit{Example 2:} Fully balancing with $\tau_i = 8$ for all $i$ violates the non-idling constraint for $i=4$, because of the late arrival of the fourth packet. Note that in Line 2 of Algorithm \ref{alg:base}, index $k=3$ yields the largest $\frac{1}{k}\sum\nolimits_{i=1}^{k} d_i$. Thus, the algorithm sets $\tau_i = \frac{4+8+18}{3} = 10$ for $i=1,2,3$, and proceeds to the next iteration. In the second iteration, $\tau_4 = d_4 = 2$ is set and the algorithm terminates.
\end{itemize}
Overall, Algorithm \ref{alg:base} seeks to balance $\tau_i$ as much as possible, while still satisfying the non-idling constraints. The non-idling constraints and the strict convexity of $w(\tau)$ jointly imply that the scheduler should maximally exploit the late arrivals of future packets (if applicable), but no more. Overall, a rule of thumb is to maximally balance transmission durations in order to minimize a strictly convex cost function, which we will also leverage to devise our proposed algorithms herein.


\section{Energy-Efficient Scheduling under Two-Sided Delay Constraints}
\label{sec:two_sided_delay}

\par In this section, we address the full problem with two-sided delay constraints. We first present the conditions that allow for a feasible solution to exist, as the two-sided problem is not guaranteed to have one. 

\subsection{Feasibility}
\label{subsec:feasibility}

\par In order to successfully transmit each packet, each packet needs to have a valid, non-negative region for its departure time (\emph{e.g.,} Figure \ref{fig:twosideddelay_demo}). In other words, the two-sided problem is infeasible whenever there exists a packet with index $i \in \{1,\dots,M\}$ whose the pre- and post-delay induced bounds do not yield a valid interval. Therefore, in order to ensure a feasible solution, a necessary condition is to have 
\begin{equation}
    \label{eq:nonfifo_feasibility}
    \begin{split}
        t_i + T_{\mathrm{pre},i} &\geq t_R - T_{\mathrm{post},i}, \quad \text{for all } i \in \{1,\dots,M\}.
    \end{split}
\end{equation}
Furthermore, FIFO ordering leads to an additional criterion for feasibility. Assume $j<i$. By definition, the post-delay-induced lower bound on the departure time of $j$ is $t_R - T_{\mathrm{post},j}$. If the pre-delay-induced departure deadline of $i$ ($t_i + T_{\mathrm{pre},i}$) is smaller than $t_R - T_{\mathrm{post},j}$, the mandatory waiting for packet $j$ violates the pre-delay constraint of packet $i$, which results in packet loss/drop. Therefore, assuming instantaneous processing with zero duration ($\tau_i = 0$) is infeasible, FIFO ordering tightens the feasibility criteria for each $i \in \{1,\dots,M\}$ to       
\begin{equation}
    \label{eq:feasibility_regions}
    \begin{split}
        t_i + T_{\mathrm{pre},i} &> t_R - T_{\mathrm{post},j} , \quad  \text{for all } j=1,\dots,i-1\\
        t_i + T_{\mathrm{pre},i} &\geq t_R - T_{\mathrm{post},i}.
    \end{split}
\end{equation}
Lastly, we need to have $T_{\mathrm{pre},i} > 0$ and $T_{\mathrm{post},i} > 0$ for all $i$ as the trivial validity conditions on pre- and post-delays. 
Throughout the paper, we assume the problem is well-posed with a feasible solution. 

\subsection{Definitions}

\par Herein, we provide some definitions to classify individual packets, as well as consecutive collections of packets. We note that these definitions will be used in the sequel to both present the scheduling algorithm, as well as prove its optimality. 
\begin{definition}[\emph{Regular end packet}]
    \label{def:regular}
    \emph{Packet $i$ is a \textit{regular end packet} if its transmission ends precisely at the arrival of the $(i+1)$th packet in the sequence. That is, a regular end packet satisfies $\sum\nolimits_{j=1}^{i}\tau_j = t_{i+1} = \sum\nolimits_{j=1}^{i} d_j$.}
\end{definition}
\begin{definition}[\emph{Pre-critical packet}]
    \label{def:precrit}
    \emph{Packet $i$ is a \textit{pre-critical packet} if its transmission ends precisely at its pre-transmission delay constraint. That is, $\sum\nolimits_{j=1}^{i} \tau_j = t_i + T_{\mathrm{pre},i}$, and the packet $i$ is said to critically satisfy its pre-delay.}
\end{definition}
\begin{definition}[\emph{Post-critical packet}]
    \label{def:postcrit}
    \emph{Packet $i$ is a \textit{post-critical packet} if its transmission ends precisely at its post-transmission delay constraint. That is, $\sum\nolimits_{j=1}^{i} \tau_j = t_R - T_{\mathrm{post},i}$ and the packet $i$ is said to critically satisfy its post-delay.}
\end{definition}
\begin{definition}[\emph{Group}]
    \label{def:group}
    \emph{A \textit{group} is the (smallest) collection of consecutive packets $\{i_1,\dots, i_2\}$ whose transmission durations satisfy $\sum\nolimits_{j=1}^{i_1-1}\tau_j = \sum\nolimits_{j=1}^{i_1-1} d_j$ and $\sum\nolimits_{j=1}^{i_2}\tau_j = \sum\nolimits_{j=1}^{i_2} d_j$. That is, $\{\tau_{i_1},\dots,\tau_{i_2}\}$ describes the time interval $[\sum\nolimits_{j=1}^{i_1-1}\tau_j,\sum\nolimits_{j=1}^{i_2}\tau_j]$ that begins and ends with a packet arrival.}
\end{definition}
Referring to Figure \ref{fig:alg_base_toyproblem}, Example 2 has two groups, $\{\tau_1,\tau_2,\tau_3\}$ and $\{\tau_4\}$, where the first group's interval ends with the arrival of the fourth packet. Example 1 has only a single group $\{\tau_1,\tau_2,\tau_3,\tau_4\}$.
\begin{definition}[\emph{Subgroup}]
    \emph{A \textit{subgroup} is a collection of consecutive packets within a group that have equal transmission durations.}
\end{definition}
\begin{definition}[\emph{Regular subgroup}]
    \emph{A \textit{regular subgroup} is a subgroup that ends with a regular end packet (\emph{i.e.,} leaves an empty buffer for the next subgroup).}
\end{definition}
\begin{definition}[\emph{Critical subgroups}]
    \emph{A \textit{pre-critical (post-critical) subgroup} is a subgroup that ends with a pre-critical (post-critical) packet.}
\end{definition}
\begin{definition}[\emph{Type-R group}]
    \emph{A \textit{regular group (type-R group)} is a group in which no pre- or post-critical packets are present.} 
\end{definition}
\begin{definition}[\emph{Type-H group}]
    \label{def:het_group}
    \emph{A \textit{heterogeneous group (type-H group)} is a group that contains at least one pre- or post-critical packet.}
\end{definition}
Note that by definition, the last subgroup of any group is a regular subgroup. Furthermore, if a group contains a single subgroup, then the group is regular (\emph{i.e.,} does not contain any delay-critical packets). We refer the reader to Figure \ref{fig:twosided_delay_example} and its discussion in Subsection \ref{subsec:twosided_energy} for exemplification on these definitions.


\subsection{Optimal Offline Scheduling under Two-Sided Constraints}
\label{subsec:twosided_energy}

\par Herein, we present an algorithm that computes the optimal offline schedule that minimizes a convex energy cost under two-sided delay constraints and solves the problem described by \eqref{eq:optimization_bothdelay}. Prior to providing the algorithm, we note for a certain packet $i$, the inter-arrival time $d_i$ may be longer than the pre-transmission delay of the packet (\textit{i.e.,} $d_i > T_{\mathrm{pre},i}$), in which case a pre-transmission delay-induced idling has to occur. We note that for a well-posed problem with a feasible solution, very similar to \cite[Proposition 3.1]{chen2008_predelay}, a case with $d_i > T_{\mathrm{pre},i}$ for some $i$ would generate two completely separated scheduling intervals (for packets $\{1,\dots,i\}$ and $\{i+1,\dots,M\}$, corresponding to time intervals $[0,t_i + T_{\mathrm{pre},i}]$ and $[t_{i+1},t_E]$), which essentially yields two independent sub-problems that do satisfy $d_i \leq T_{\mathrm{pre},i}$ for all $i$.\footnote{Note that if there exists multiple such packets, this argument can be recursively generalized.} Since solving these two sub-problems separately is equivalent to finding the optimal $\boldsymbol{\tau}$ of the overall system, we assume $d_i \leq T_{\mathrm{pre},i}$ for all $i$ throughout the rest of the paper. 

\begin{algorithm}[!t]
	\fontsize{10}{10}\selectfont
	\begin{algorithmic}[1]
			\caption{Optimal (energy-minimizing) offline scheduling for two-sided transmission delays.}
			\label{alg:general_twosided}
			\renewcommand{\algorithmicrequire}{\textbf{Inputs:}}
			\renewcommand{\algorithmicensure}{\textbf{Output:}}
			\REQUIRE $M$, $\{T_{\mathrm{post},1}, \dots , T_{\mathrm{post},M} \}$, $\{T_{\mathrm{pre},1}, \dots , T_{\mathrm{pre},M} \}$, $\boldsymbol{t}$, $t_R$ \\
			Set pre-delay-induced departure deadlines $\nu_i = t_i + T_{\mathrm{pre},i}$, for  $i=1,\dots,M$\\
			Set end time $t_E = t_M + T_{\mathrm{pre},M}$ \mcg{and initialize $t_{R,o} = t_R$}\\
			\dotfill \\
			For the first packet: \\
			\STATE $\tau_{[1]} = \max \{d_1 , \mcg{t_{R,o}} - T_{\mathrm{post},1}\}$, $\quad n_{[1]} \doteq 1$ \\
            \STATE $\tau_{[2]} = \min \big\{T_{\mathrm{pre},1}, \max \{\frac{d_1+d_2}{2}, \frac{\mcg{t_{R,o}} - T_{\mathrm{post},2}}{2} \} \big\}$ \\
            $n_{[2]} = \argmin \left\{T_{\mathrm{pre},1}, \max \{\frac{d_1+d_2}{2}, \frac{\mcg{t_{R,o}} - T_{\mathrm{post},2}}{2} \} \right\}$ \\
            \STATE $A_{[k]} = \Big\{T_{\mathrm{pre},1}, \frac{d_1+T_{\mathrm{pre},2}}{2} \dots , \frac{T_{\mathrm{pre},k-1} + \sum_{i=1}^{k-2} d_i}{k-1}, $ \\
                \hspace{3.8cm} $\max \{\frac{\sum_{i=1}^{k} d_i}{k},  \frac{\mcg{t_{R,o}} - T_{\mathrm{post},k}}{k} \} \Big\} $ \\
            $\tau_{[k]} = \min A_{[k]}$ \\
            $n_{[k]} = \argmin A_{[k]}$ \\
            \STATE $i^* = \argmax \limits_{i} \tau_{[i]}$  \\ 
            $\tau_1 \leftarrow \tau_{[i^*]}$ \\
            $\tau_1 = \dots = \tau_{n_{[i^*]}} \leftarrow \tau_1$ \\
            \dotfill \\
            For the next iteration(s): \\
            \IF{$n_{[i^*]} = i^*$ and $\sum_{i=1}^{i^*} d_i \geq \mcg{t_{R,o}} - T_{\mathrm{post},i^*}$}
                \STATE \%\%The last packet in the allocated batch is regular. \\
                \STATE Move to packet $n_{[i^*]} + 1 = i^* + 1$ \mcg{with $\mcg{t_{R,o}} \leftarrow t_R - t_{i^* + 1}$}, and proceed the same as Lines $1$-$4$.
            \ELSIF{$n_{[i^*]} = i^*$ and $\sum_{i=1}^{i^*} d_i < \mcg{t_{R,o}} - T_{\mathrm{post},i^*}$} 
                \STATE \%\%The last packet in the allocated batch is post-critical.
                \FOR{$i = i^*+1$ to $M$}
                    \STATE Set $t_i = \max \{ t_R - T_{\mathrm{post},i^*} ,t_i \}$ \\
                    \IF{$t_i = t_R - T_{\mathrm{post},i^*}$}
                        \STATE $T_{\mathrm{pre},i} \leftarrow \nu_i - (t_R - T_{\mathrm{post},i^*})$ 
                    \ENDIF
                \ENDFOR
                \STATE Update $d_i = t_{i+1} - t_i$ for $i=i^*+1,\dots,M$ \\ 
                \STATE Move to packet $n_{[i^*]} + 1 = i^* + 1$ \mcg{with $\mcg{t_{R,o}} \leftarrow T_{\mathrm{post},i^*}$}, and proceed the same as Lines $1$-$4$.
            \ELSIF{$n_{[i^*]} \neq i^*$}
                \STATE \%\%The last packet in the allocated batch is pre-critical. 
                \FOR{$i = n_{i^*}+1$ to $M$}
                    \STATE Set $t_i = \max \{ \nu_{n_{i^*}} ,t_i \}$ \\
                    \IF{$t_i = t_{n_{i^*}} + T_{\mathrm{pre},n_{i^*}}$}
                        \STATE $T_{\mathrm{pre},i} = \nu_i - \nu_{n_{i^*}}$ 
                    \ENDIF
                \ENDFOR
                \STATE Update $d_i = t_{i+1} - t_i$ for $i= n_{i^*+1},\dots,M$ \\ 
                \STATE Move to packet $n_{[i^*]} + 1$  \mcg{with $\mcg{t_{R,o}} \leftarrow t_R - \nu_{n_{i^*}}$}, and proceed the same as Lines $1$-$4$.
            \ENDIF
 			\STATE \textbf{return} $\boldsymbol{\tau}$ 
	\end{algorithmic}
\end{algorithm} 

\par Algorithm \ref{alg:general_twosided} presents the proposed optimal offline scheduling algorithm for the two-sided problem defined by \eqref{eq:optimization_bothdelay}. Similar to the single deadline scenario in Section \ref{sec:prelims}, Algorithm \ref{alg:general_twosided} also seeks to balance the transmission durations as much as possible. However, the natural differences in the two-sided delay scenario are the strict limits induced by the pre- and post-transmission delay constraints. \mcg{We first note that the numerators of every candidate term in $A_{[k]}$ (for definition, see algorithm body) correspond to the pre-delay-induced departure deadlines of each packets up to $k$ (except the last). Thus, the $\min$ operation executed when computing $\tau_{[k]}$ selects the most restrictive pre-delay constraint (say $i$th) and allocates equally for the first $i$ packets. Using this design in every iteration circumvents any pre-delay violations. The last term in $A_{[k]}$ is a $\max$ operation among the $(k+1)$th packet's arrival time and the $k$th packet's post-delay constraint, and ensures no idle time and post-delay violation at packet $k$, in case no pre-delay is further constraining the system (see discussion above). This, alongside the overall $\argmax$ operation in Line 4, together avoid post-delay constraint violations, implying that Algorithm \ref{alg:general_twosided} provides an algorithm that yields a valid solution.}

\begin{figure*}[!t]
    \centering
    \includegraphics[width = 0.95\textwidth]{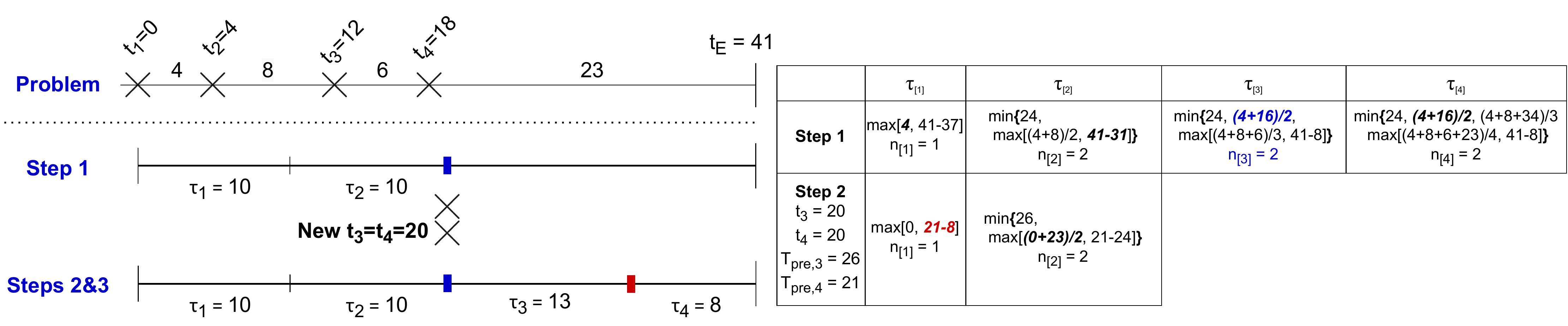}
    \caption{\mcg{A demonstration of Algorithm \ref{alg:general_twosided}, with $M=4$ packets. The pre- and post-delay constraints $\boldsymbol{T}_{\mathrm{pre}} = [24,16,34,23]$ and $\boldsymbol{T}_{\mathrm{post}} = [37,31,8,24]$, and the reference time $t_R = 41$. The table on the right hand side demonstrates the algorithm's operations in a step by step manner. Blue/red color represent pre-/post-criticality of the marked packet, respectively. }}
    \label{fig:twosided_delay_example}
\end{figure*}

\par Herein, we present a toy problem in Figure \ref{fig:twosided_delay_example} to demonstrate the operating principles of the algorithm\footnote{Unlike the single delay scenario, $t_E < t_R$ is possible under two-sided deadlines, in cases where $T_{\mathrm{pre},M} < t_R - t_M$.}. We first note that although equalizing all $\tau_i$ is optimal to minimize a convex cost, it is not feasible due to multiple pre-/post-delay constraints. In particular, the pre-delay constraint of the second packet strictly upper bounds the departure of the packet by $\tau_1 + \tau_2 \leq t_2 + T_{\mathrm{pre},2} = 20$. Allocating these two packets equally within said interval is indeed valid for both their pre-/post-delays, thus the algorithm allocates $\tau_1 = \tau_2 = 10$ in its first iteration. Note that in the meantime, packets $3$ and $4$ have arrived, but are waiting in queue for processing due to FIFO. At time $t=20$, the allocation of these packets begins. Due to convexity of $w(\tau)$, the remaining time of $21$ units is best allocated by equalizing $\tau_3$ and $\tau_4$, however this is infeasible due to the late post-delay of packet $3$. To satisfy $T_{\mathrm{post},3} = 10$, the algorithm allocates $\tau_3 = 13$, and the remaining duration is left to the last packet $\tau_4 = 8$. \mcg{The steps taken by the algorithm are also presented on Figure \ref{fig:twosided_delay_example} for demonstrative purposes.} Overall, the resultant $\boldsymbol{\tau}$ for this toy example consists of a single type-H group, with three subgroups: The first subgroup $\{\tau_1,\tau_2\}$ is pre-critical, the second subgroup $\{\tau_3\}$ is post-critical, and the last one $\{\tau_3\}$ is regular. Furthermore, we observe that in cases where pre- or post-delays are constraining fully balancing the durations, Algorithm \ref{alg:general_twosided} allocates the durations critically. That is, if the constraining deadline is a pre-delay of packet $i$, then $\sum_{j=1}^{i} \tau_j = t_i + T_{\mathrm{pre},i}$, or if it is a post-delay of packet $i'$, then $\sum_{j=1}^{i'} \tau_j = t_R - T_{\mathrm{post},i'}$. 

\subsection{Optimality of Algorithm \ref{alg:general_twosided}}

\par Prior to proving the output of Algorithm \ref{alg:general_twosided} is the optimal one, we first provide Lemma \ref{lem:optimality_conds_subgroup}. 
\begin{lemma}
    \label{lem:optimality_conds_subgroup}
    For an optimal allocation $\boldsymbol{\tau}$ for the two-sided delay case, we have the following relations:
    \begin{enumerate}
        \item A regular subgroup has a transmission duration that is {\bf greater than or equal to} the subsequent subgroup. 
        \item A post-critical subgroup has a transmission duration that is {\bf greater than or equal to} the subsequent subgroup. 
        \item A pre-critical subgroup has a transmission duration that is {\bf less than or equal to} the subsequent subgroup.
    \end{enumerate}
\end{lemma}
\begin{proof}
    The first two items follow from Lemma \ref{lem:decreasing_tau_postdelay} and \cite[Lemma 2]{uysal2002_singledelay}.
    Note that if they did not hold, then one could simply yield a smaller energy cost by further \emph{balancing} the transmission durations through decreasing the durations of the subsequent subgroup and increasing the durations of the subgroup of interest, without violating the delay constraints (similar to Lemma \ref{lem:decreasing_tau_postdelay}). Similarly, if the third condition did not hold, the scheduler could decrease the cost by further balancing, via increasing the durations of the subsequent subgroup and decreasing the durations of the subgroup of interest, again without violating the (pre-)delay constraints (also see \cite[Lemma 3.2]{chen2008_predelay}).\footnote{Though it is unlikely to occur in practice, the problem formulation theoretically allows for cases with coinciding pre-/post-delay induced bounds on departure times, that is $t_i + T_{\mathrm{pre},i} = t_R - T_{\mathrm{post},i}$ for some $i$. This forces the departure time of packet $i$ to a single point, essentially dividing the problem into two sub-problems for packets $\{1,\dots,i\}$ and $\{i+1,\dots,M\}$. In such cases, the pre-/post-criticality of the subgroup ending with $i$ can be determined in compliance with this lemma.}
\end{proof}
Herein we present Lemma \ref{lem:ouralg_decreasingtau}, which shows that the regular and post-critical subgroups produced by Algorithm \ref{alg:general_twosided} satisfies the comparative relations dictated by Lemma \ref{lem:optimality_conds_subgroup} (conditions 1 and 2). The proof for a pre-critical subgroup follows an identical procedure, and is omitted for conciseness.
 \begin{lemma}
     \label{lem:ouralg_decreasingtau}
     Let $\boldsymbol{\tau}$ denote the output of Algorithm \ref{alg:general_twosided}. For a pair of consecutive subgroups in $\boldsymbol{\tau}$ where the prior subgroup is regular or post-critical, let $\tau_1$ and $\tau_2$ be their transmission durations, respectively. Then, $\tau_1 \geq \tau_2$.
 \end{lemma}
 \begin{proof}
     Let the algorithm's output be such that the former subgroup has $n_1$ packets of $\tau_1$ duration, and the latter subgroup has $n_2$ packets with $\tau_2$ duration. Suppose we have $\tau_1 < \tau_2$. Then, the total time interval covered by both subgroups, $(n_1 \tau_1 + n_2 \tau_2)$, satisfies
     \begin{equation}
         \label{eq:lemma_dectau_dummy}
         \begin{split}
             \frac{n_1 \tau_1 + n_2\tau_2}{n_1+n_2} &= \frac{n_1}{n_1+n_2} \tau_1 + \frac{n_2}{n_1+n_2} \tau_2 \\ &> \frac{n_1}{n_1+n_2} \tau_1 + \frac{n_2}{n_1+n_2} \tau_1 =\tau_1.
         \end{split}
     \end{equation}
     Let the index of the first packet of the first subgroup have the index $(k+1)$. Equation \eqref{eq:lemma_dectau_dummy} implies that from the algorithm's perspective, allocating $\frac{n_1 \tau_1 + n_2\tau_2}{n_1+n_2}$ for packets $\{k+1,\dots,k+n_1+n_2\}$ would yield a larger value than allocating $\tau_1$ for packets $\{k+1,\dots,k+n_1\}$. Thus, given the option of allocating for packets $\{k+1,\dots,k+n_1+n_2\}$ a duration of $\frac{n_1 \tau_1 + kn_2\tau_2}{n_1+n_2}$, the maximizer at Line 4 of Algorithm \ref{alg:general_twosided} cannot choose allocating $\tau_1$ for packets $\{k+1,\dots,k+n_1\}$ (\emph{i.e.,} there is a larger value\footnote{For simplicity of presentation, this argument implicitly assumes such an increase does not violate any pre-delays. That said, if there is a pre-delay constraint that prevents this, a critical allocation for said pre-delay constraint also satisfies the same argument.}), which yields a contradiction. Thus, consecutive subgroups have to have non-increasing durations.
 \end{proof}
 
Overall, in light of Lemmas \ref{lem:optimality_conds_subgroup}-\ref{lem:ouralg_decreasingtau}, we provide the main result for the two-sided case in the theorem below.
\begin{theorem}
    \label{theo:optimality_twosided}
    Algorithm \ref{alg:general_twosided} yields an optimal $\boldsymbol{\tau}$ that minimizes any cost that is strictly convex, decreasing, and positive.
\end{theorem}
\begin{proof}
    See Appendix \ref{ap:proof_of_twosided}.
\end{proof}
We emphasize that similar to the single deadline case (\cite{uysal2002_singledelay}) and the pre-delay only case (\cite{wanshi_multihop}), the optimal solution for the two-sided case is also independent of the particular energy function $w(\tau)$. That is, regardless of the particular function, as long as $w(\tau)$ is strictly convex and decreasing, the output of Algorithm \ref{alg:general_twosided} is the optimal schedule. 

\par A special case of the scheduling problem under two-sided delay constraints is the case wherein only post-delays are relevant, and pre-delays are $T_{\mathrm{pre},i} = t_R - t_i$ for all $i$ (thus unconstraining). For this so-called \emph{one-sided}, post-delay constrained case, we have the following lemma: 
\begin{lemma}
    \label{lem:decreasing_tau_postdelay}
    Having $\tau_i \geq \tau_{i+1}$ for all $i$ is a necessary condition for optimality in the one-sided, post-transmission delay-constrained scenario.
\end{lemma}
\noindent The proof is identical to \cite[Lemma 2]{uysal2002_singledelay}. \mcg{Lemma \ref{lem:decreasing_tau_postdelay} extends the necessary condition of having non-increasing durations in the optimal $\boldsymbol{\tau}$ of the single deadline case in \cite{uysal2002_singledelay} to systems under one-sided, individual post-transmission delay constraints.} In the sequel, the lemma will motivate our design of the transmission completion time minimizing algorithm (Subsection \ref{susbec:converse_postdelay}).

\noindent \mcg{\textbf{Computational Complexity:} The computational complexity of Algorithm \ref{alg:general_twosided} greatly depends on how many iterations are needed to fully allocate all $M$ packets. At the worst case, the allocation occurs one packet at a time, and the algorithm takes $M$  iterations to complete. In its first iteration, Algorithm \ref{alg:general_twosided} compares $M$ candidate $\tau_{[k]}$ values to come up with the first allocation, which requires evaluating $k$ elements of the set $A_{[k]}$ (see definition in algorithm body). This incurs computing $1+ \dots + M = \frac{M(M+1)}{2}$ elements and is $\mathcal{O}(M^2)$. Rest of the pre-/post-delay updates in Lines 5-28 are $\mathcal{O}(M)$. Therefore, each iteration of Algorithm \ref{alg:general_twosided} is $\mathcal{O}(M^2)$, implying the overall time complexity is $\mathcal{O}(M^3)$ in the worst case. 


\par Storing and updating every $A_{[k]}$ set incurs $\mathcal{O}(M^2)$ space complexity (similar to the paragraph above). As storing $\boldsymbol{d}$, $\boldsymbol{T}_{\mathrm{pre}}$, $\boldsymbol{T}_{\mathrm{post}}$, are all $\mathcal{O}(M)$, Algorithm \ref{alg:general_twosided} incurs $\mathcal{O}(M^2)$ space complexity. However, we note that as $\tau_{[k]}$ and $n_{[k]}$ values are computed by $\max$ and $\min$ operations on $A_{[k]}$, each element of $A_{[k]}$ can be fed one-by-one for comparison, in which case the space complexity reduces to $\mathcal{O}(M)$.
}

\section{Energy Constrained Completion Time Minimization}
\label{sec:completiontime}

\par Section \ref{sec:two_sided_delay} introduces strategies that minimize the incurred cost (\emph{e.g.,} energy consumption) while subjecting each packet to pre- and post-transmission delay constraints. However, the ``dual" problem may also be of interest, where the scheduler seeks to minimize the total transmission completion time $T_c=\sum_{i=1}^{M} \tau_i$, with a fixed energy budget $w_{\mathrm{max}}$. 
Herein, we seek to find the optimal offline allocation that yields said shortest transmission completion time, while transmitting all $M$ packets until a reference time $t_R$, as well as having all $M$ packets received and active (\emph{i.e.,} not expired) at $t_R$ (similar to Sections \ref{sec:system_model}-\ref{sec:two_sided_delay}). In addition to an energy budget constraint $w_{\mathrm{max}}$, this problem setup again yields \emph{pre-delay} and \emph{post-delay constraints} as well. The overall problem can be formulated as
\begin{equation}
	\begin{aligned}
		\min_{\boldsymbol{\tau}} \quad & T_c = \sum\nolimits_{i=1}^{M} \tau_i\\	
		\text{s.t.} \quad & \sum\nolimits_{i=1}^{k} \tau_{i} \geq \sum\nolimits_{i=1}^{k} d_{i}, \quad k \in \{1, \dots, M-1\},  \\ 
		& w(\boldsymbol{\tau}) = \sum\nolimits_{i=1}^{M} w(\tau_i) \leq w_{\mathrm{max}}, \\
		& \sum\nolimits_{i=1}^{k} \tau_{i} \geq t_R - T_{\mathrm{post},k}, \quad k \in \{1, \dots, M\}, \\
		& \sum\nolimits_{i=1}^{k} \tau_{i}-\sum\nolimits_{i=1}^{k-1} d_{i} \leq T_{\mathrm{pre},k}, \hspace{0.1cm} k \in \{1, \ldots, M \}.
	\end{aligned}
    \label{eq:converse_problem}
\end{equation}

\subsection{Completion Time Minimization under Post-Delays}
\label{susbec:converse_postdelay}

\par Though the ultimate goal is addressing the two-sided problem (presented in the next subsection), we note that pre- and post-delays need to be treated differently for the completion time minimization framework. To provide insight into the role of post-delays in the completion time problem, as well as build the results in a cumulative manner, we first address the one-sided, post-delay only case. Recall that the post-delay only case corresponds to having 
$T_{\mathrm{pre},i} = t_R-t_i$ for all $i$. Furthermore, note that depending on the relationship between $t_R$, $w(\tau)$, $w_{\max}$, $\boldsymbol{d}$, and $\boldsymbol{T}_{\mathrm{post}}$, there can be scenarios in which $T_c \leq t_R$ is not feasible (\textit{e.g.,} not enough energy to complete transmissions until $t_R$). Herein, we assume the problem is well-posed and a feasible solution exists.

\par To motivate algorithm design, we provide the following lemma that provides necessary conditions for the optimal duration vector $\boldsymbol{\tau}$ that achieves the minimum completion time under two-sided delays (denoted by $T^*_c$ throughout the paper). Note that if the strict post-delay lower bound of packet $M$, $T_c = t_R - T_{\mathrm{post},M}$, is achievable with the available budget $w_{\mathrm{max}}$, all $\boldsymbol{\tau}$ that satisfy $w(\boldsymbol{\tau}) \leq w_{\mathrm{max}}$ and $\sum_{i=1}^{M} \tau_i = t_R - T_{\mathrm{post},M}$ yield the smallest possible completion time, making the smallest completion time $T^*_c = t_R - T_{\mathrm{post},M}$. When presenting Lemmas \ref{lem:converse_needs_energy_optimal}-\ref{lem:converse_converse}, we will consider the more general case where $T^*_c > t_R - T_{\mathrm{post},M}$.

\begin{lemma}[Necessary conditions for completion time optimality]
    \label{lem:converse_needs_energy_optimal}
    Let $T^*_c > t_R - T_{\mathrm{post},M}$ be the minimum completion time, achieved by some duration vector $\boldsymbol{\tau}$. Then, the following are true:
    \begin{enumerate}
        \item $w(\boldsymbol{\tau}) = w_{\mathrm{max}}$.
        \item The vector $\boldsymbol{\tau}$ is the optimal vector for the energy minimization problem with end time $t_E = T^*_c$, and its achieved minimum energy is $w_{\mathrm{max}}$.
    \end{enumerate}
\end{lemma}
\begin{proof}
    We show the statements in their order of presentation:
    \begin{enumerate}
        \item Suppose $w(\boldsymbol{\tau}) < w_{\mathrm{max}}$. Then, $\boldsymbol{\tau}$ leaves unused energy that could have been used to further decrease $\tau_M$, thereby decreasing $T_c$. Hence, $T^*_c < \sum_{i=1}^{M} \tau_i$, yielding a contradiction. Thus, $w(\boldsymbol{\tau}) = w_{\mathrm{max}}$ is true.
        \item Suppose otherwise. Then, there exists a feasible $\boldsymbol{\tau}' \neq \boldsymbol{\tau}$ such that $\sum_{i=1}^{M} \tau'_i = \sum_{i=1}^{M} \tau_i = T^*_c$, that achieves $w(\boldsymbol{\tau}') < w(\boldsymbol{\tau})$. Thus, $w(\boldsymbol{\tau}') < w_{\mathrm{max}}$, meaning $T^*_c$ cannot be the minimum completion time due to Condition 1, yielding a contradiction.
    \end{enumerate}
\end{proof}

\begin{corollary}
    \label{cor:converse_post_decreasing}
    For any $T^*_c > t_R - T_{\mathrm{post},M}$, the optimal $\boldsymbol{\tau}$ that achieves $T^*_c = \sum_{i=1}^{M} \tau_i$ for the post-delay only case satisfies $\tau_i \geq \tau_{i+1}$ for all $i$.
\end{corollary}
\begin{proof}
    Follows directly from Lemmas \ref{lem:decreasing_tau_postdelay} and \ref{lem:converse_needs_energy_optimal}.
\end{proof}

\begin{algorithm}[!t] 
	\fontsize{10}{10}\selectfont
	\begin{algorithmic}[1]
			\caption{Offline scheduling algorithm for transmission completion time minimization under post-transmission delay constraints.}
			\label{alg:converse_post}
			\renewcommand{\algorithmicrequire}{\textbf{Inputs:}}
			\renewcommand{\algorithmicensure}{\textbf{Output:}}
			\REQUIRE $M$, $\boldsymbol{d}$, $w(\tau)$, $t_R$, $\boldsymbol{T}_{\mathrm{post}}$, and for all packets $T_{\mathrm{pre},i} = t_R - t_i$  \\

			\STATE $i^* = \argmax\limits_{i} \quad t_R - T_{\mathrm{post},i} $ \\
			
			\IF{$t_R - T_{\mathrm{post},i^*} \leq t_M$}
			    \STATE \%\% \textit{Case 1: All post-delay deadlines occur before the last packet's arrival.} \\
    			\STATE Run Algorithm \ref{alg:general_twosided} for packets $1,\dots,M-1$ with inter-arrival times $\{d_1,\dots,d_{M-1}\}$. Obtain $\{\tau^{(b)}_1,\dots,\tau^{(b)}_{M-1}\}$\\
        		$n_g$: the number of subgroups. \\
        		$\boldsymbol{n}^{(m)}_{\hspace{0.1cm} n_g \times 1}$: vector that holds the number of packets in each subgroup. \\
        		\STATE $k = \sum_{i=1}^{n_g} n^{(m)}_{i} = M-1$
        		
        		\WHILE{1}
        		    \STATE Set $\tau_i = \tau^{(b)}_i$, for $i=1,\dots,k$ \\ 
        		    \STATE Compute remaining energy $w_r = w_{\mathrm{max}} - \sum_{i=1}^{k} w(\tau^{(b)}_i)$ \\
        		    \STATE Assign equal duration for the remaining packets: $\tau_{k+1} = \dots = \tau_M = w^{-1}(\frac{w_r}{M-k})$ \\
        		    
        		    \IF{$\tau_i \geq \tau_{i+1}$ for all $i$}
        		        \STATE \textbf{break}
        	        \ELSE
        	            \STATE $n_g = n_g -1$\\
        	            \STATE $k = \sum_{i=1}^{n_g} n^{(m)}_{i} (n_g)$
        		    \ENDIF
        		
        		\ENDWHILE
			
			\ELSE
			    \STATE \%\% \textit{Case 2: There exists a post-delay deadline after the last packet's arrival.} \\
			    \IF{$i^* \neq M$}
			        \STATE \%\% \textit{Case 2a: Although occurring after $t_M$, the latest post delay is one of the earlier packets.}
			        \STATE Run Algorithm \ref{alg:general_twosided} for packets $1,\dots,i^*$, with end time $t_R-T_{\mathrm{post},i^*}$. Obtain $\tau^{(b)}_1,\dots,\tau^{(b)}_{i^*}$ \\
			        \STATE Proceed as in Lines 6-16, with initializing $k=i^*$ and $n_g$ as the subgroup index of packet $i^*$.
		        \ELSIF{$i^* = M$}
			        \STATE Run Algorithm \ref{alg:general_twosided} for packets $1,\dots,M$, with end time $t_R-T_{\mathrm{post},M}$. Obtain $\tau_1,\dots,\tau_{M}$ \\
			        \IF{$w(\boldsymbol{\tau}) \leq w_{\mathrm{max}}$}
			            \STATE Minimum possible completion time $t_R-T_{\mathrm{post},M}$ is achievable. Algorithm done. \textit{\%\% Case 2b-i}\\
			        \ELSIF{$w(\boldsymbol{\tau}) > w_{\mathrm{max}}$}
			            \STATE \textit{\%\% Case 2b-ii}
			            \IF{$\forall i \in \{1,\dots,M-1\}$, we have $t_R - T_{\mathrm{post},i} < t_M$}
			                \STATE Proceed as in Case 1.
			            \ELSIF{$\exists i \in \{1,\dots,M-1\}$, such that $t_R - T_{\mathrm{post},i} \geq t_M$}
    			            \STATE Proceed as in Case 2a, with $i^*$ equal to the index of the second largest $t_R - T_{\mathrm{post},i}$.
			            \ENDIF
			        \ENDIF
			    \ENDIF
			\ENDIF \\
			\textbf{return} $\boldsymbol{\tau}$
	\end{algorithmic}
\end{algorithm} 

\par Motivated by Lemma \ref{lem:converse_needs_energy_optimal} and Corollary \ref{cor:converse_post_decreasing}, we present Algorithm \ref{alg:converse_post} that achieves the minimum transmission completion time. Note that according to the post-delay-induced departure time lower bounds, the approach taken by Algorithm \ref{alg:converse_post} changes, hence the algorithm is presented in a case-by-case manner. Herein, to describe the operations of Algorithm \ref{alg:converse_post}, we address Case 1, where all post-delay induced lower limits on departure times reside within the interval $[0,t_M)$. We note that the rest of the cases (\emph{e.g.,} Case 2 and its subcases) follow a similar logic, with the exception of accounting for post-delays occurring after $t_M$. 

\par For Case 1, Algorithm \ref{alg:converse_post} exploits the fact that the completion time $T_c$ cannot be made smaller than the arrival time of the last packet due to causality. Thus, if one were to allocate the first $M-1$ packets between the interval $[0,t_M]$ via Algorithm \ref{alg:general_twosided} such that their incurred energy consumption is minimized, the largest possible energy would be left for the last packet's allocation. Given $w(\tau)$ is decreasing in $\tau$, the largest energy budget would yield the shortest possible $\tau_M$, thereby minimizing $T_c = \sum_{i=1}^{M} \tau_i$. 

\par Now suppose that using the construction in the above paragraph, the resultant $\boldsymbol{\tau}$ yields $\tau_M > \tau_{M-1}$. Then, using the argument in Corollary \ref{cor:converse_post_decreasing}, such a $\boldsymbol{\tau}$ cannot be optimal. To handle this case, we first note that Algorithm \ref{alg:general_twosided} is employed for packets $[1,\dots,M-1]$ and outputs transmission durations that come in subgroups (say its last subgroup is \mcg{$\{\tau_{k_e+1},\dots,\tau_{M-1}\}$}). We know from the optimality of Algorithm \ref{alg:general_twosided}, that the packets $\{1,\dots,k_e\}$ are allocated such that their energy consumption is minimized within the interval $[0,\sum_{i=1}^{k_e} \tau_i]$. Thus, the remaining energy $w_{\mathrm{max}} - \sum_{i=1}^{k_e} w(\tau_i)$ to be allocated for the durations $\{\tau_{k_e+1},\dots,\tau_{M}\}$ is as large as possible. Since the energy function $w(\tau)$ is strictly convex and decreasing in $\tau$, in cases where $\tau_M > \tau_{M-1}$ in the initial attempt, Algorithm \ref{alg:converse_post} adds another iteration and equalizes the transmission durations of packets $\{k_e+1,\dots,M\}$, so that the sum duration is minimized. Using this principle, Algorithm \ref{alg:converse_post} runs in a backwards manner, subgroup by subgroup, until the requirement of non-increasing $\tau_i$ required by Corollary \ref{cor:converse_post_decreasing} is met. Figure \ref{fig:converse_postonly} presents an illustrative example for the argument made herein. Overall, the principle of Algorithm \ref{alg:converse_post} relies on 
\begin{itemize}
    \item partitioning the set of packets into two subsets $\{1,\dots,k\}$ and $\{k+1,\dots,M\}$ for some $k$,
    \item performing energy minimization on the left hand side, and
    \item using the maximized remaining energy to minimize completion time by equalizing durations on the right hand side.
\end{itemize}
Here, as also discussed in the above paragraph, determining the index of separation ($k$) is motivated by Corollary \ref{cor:converse_post_decreasing}.
\begin{figure*}[!t]
    \centering
    \includegraphics[width = 0.48\textwidth]{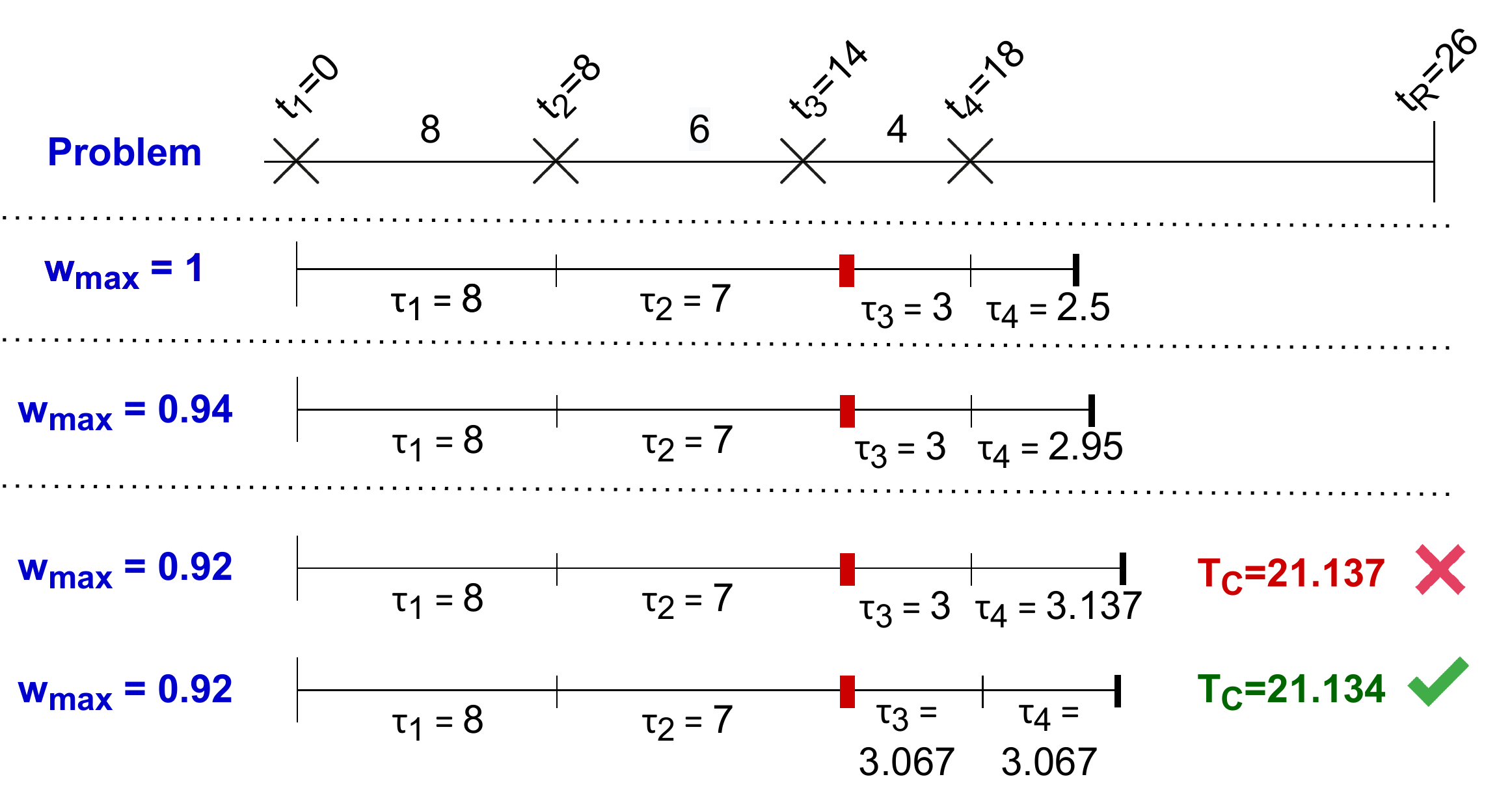}
    \caption{Demonstrative toy example for Algorithm \ref{alg:converse_post} under different $w_{\mathrm{max}}$ for $\boldsymbol{T}_{\mathrm{post}} = [\infty,11,\infty,\infty]$ and $t_R=26$. The energy cost function $w(\tau)=\frac{1}{\tau}$. Blue/red marks on departure times represent pre-/post-criticality of the marked packet, respectively.}
    \label{fig:converse_postonly}
\end{figure*}

\par Prior to proving the optimality of Algorithm \ref{alg:converse_post}, we provide the following lemma that gives sufficient conditions for optimality. Note that similar to Lemma \ref{lem:converse_needs_energy_optimal}, we again consider the more general case of $T^*_c > t_R - T_{\mathrm{post},M}$. 
\begin{lemma}[Converse of Lemma \ref{lem:converse_needs_energy_optimal}]
    \label{lem:converse_converse}
    If a vector $\boldsymbol{\tau}$ minimizes the energy cost within the time interval $[0,\sum_{i=1}^{M} \tau_i]$ with the achieved minimum energy cost $w(\boldsymbol{\tau}) = w_{\mathrm{max}}$, then it is completion-time-optimal with $\sum_{i=1}^{M} \tau_i = T^*_c$.
\end{lemma}
\begin{proof}
    Suppose $\sum_{i=1}^{M} \tau_i$ is energy-optimal in $[0,\sum_{i=1}^{M} \tau_i]$ with $w(\boldsymbol{\tau}) = w_{\mathrm{max}}$, but the optimal duration vector is $\boldsymbol{\tau}' \neq \boldsymbol{\tau}$, that is $\sum_{i=1}^{M} \tau'_i = T^*_c < \sum_{i=1}^{M} \tau_i$. From Lemma \ref{lem:converse_needs_energy_optimal}, we necessarily have $w(\boldsymbol{\tau}') = w_{\mathrm{max}}$. Then, due to $w(\tau)$ being decreasing in $\tau$, by increasing only $\tau'_M$ so that $\sum_{i=1}^{M} \tau'_i = \sum_{i=1}^{M} \tau_i$ is met, one can achieve $w(\boldsymbol{\tau}') < w_{\mathrm{max}}$. Hence, $\boldsymbol{\tau}'$ cannot be the energy-minimizing duration vector for the interval $[0,T_c]$, which yields a contradiction. 
\end{proof}
\begin{theorem}
    \label{theo:converse_post_optimal}
    The output of Algorithm \ref{alg:converse_post} (say $\boldsymbol{\tau}$) yields the optimal transmission duration vector that minimizes $T_c$. 
\end{theorem}
\begin{proof}
    The proof is presented in Appendix \ref{ap:proof_of_comptime_post}.
\end{proof}    
Note that unlike the energy-minimization problem, the values of the optimal duration vector $\boldsymbol{\tau}$ changes with the particular energy function $w(\tau)$. That said, as long as $w(\tau)$ is strictly convex and decreasing, the optimal $\boldsymbol{\tau}$ is guaranteed to be found by Algorithm \ref{alg:converse_post}.


\subsection{Completion Time Minimization under Two-Sided Constraints}
 
\par Herein, we extend the approach for the post-delay only case to the more generalized, two-sided case. The algorithm that yields the optimal $\boldsymbol{\tau}$ is presented in Algorithm \ref{alg:converse_twosided}. Note that when presenting Algorithm \ref{alg:converse_twosided}, we only provided the explicit algorithm for the case where all post-delay deadlines occur before the last packet's arrival (\emph{i.e.,} corresponds to Case 1 of Algorithm \ref{alg:converse_post} for the post-delay only scenario), for brevity. We emphasize that similar to Case 1, Cases 2a and 2b are both identical extensions of their post-delay only counterparts, with the only difference arising when handling pre-delay violations. For all cases, said pre-delay violations are again handled as described in Lines 23-35 of Algorithm \ref{alg:converse_twosided}.


\begin{algorithm}[!t] 
	\fontsize{10}{10}\selectfont
	\begin{algorithmic}[1]
			\caption{Optimal offline scheduling for transmission completion time minimization under two-sided delay constraints.}
			\label{alg:converse_twosided}
			\renewcommand{\algorithmicrequire}{\textbf{Inputs:}}
			\renewcommand{\algorithmicensure}{\textbf{Output:}}
			\REQUIRE $M$, $\boldsymbol{d}$, $w(\tau)$, $t_R$, $\boldsymbol{T}_{\mathrm{post}}$, $\boldsymbol{T}_{\mathrm{pre}}$  \\
			
			
		    \STATE \%\% \textit{Case 1: All post-delay deadlines occur before the last packet's arrival ($t_R - T_{\mathrm{post},i} \leq t_M $).} \\
			\STATE Run Algorithm \ref{alg:general_twosided} for packets $1,\dots,M-1$ with inter-arrival times $[d_1,\dots,d_{M-1}]$. Obtain $\tau^{(b)}_1,\dots,\tau^{(b)}_{M-1}$\\
    		$n_g$: the number of subgroups. \\
    		$\boldsymbol{n}^{(m)}_{\hspace{0.1cm} n_g \times 1}$: vector that holds the number of packets in each subgroup. \\
    		\STATE $k = \sum_{i=1}^{n_g} n^{(m)}_{i} = M-1$
    		
    		\WHILE{1}
    		    \STATE Set $\tau_i = \tau^{(b)}_i$, for $i=1,\dots,k$ \\ 
    		    \STATE Compute remaining energy $w_r = w_{\mathrm{max}} - \sum_{i=1}^{k} w(\tau^{(b)}_i)$ \\
    		    \STATE Assign equal duration for the remaining packets: $\tau_{k+1} = \dots = \tau_M = w^{-1}(\frac{w_r}{M-k})$ \\
    		    \STATE Check among $[k+1,\dots,M]$ for pre-delay violations
    		    \IF{$\sum_{j=1}^{i} \tau_j \leq t_i + T_{\mathrm{pre},i}$ for all $i=k+1,\dots,M$}
    		        \STATE \%\% \textit{No pre-delay violations, check for non-increasing ordering in the last subgroup}
        		    \IF{$\tau_k \geq \tau_{k+1}$}
        		        \STATE \textbf{break}
        	        \ELSE
        	            \STATE $n_g = n_g -1$\\
        	            \STATE $k = \sum_{i=1}^{n_g} n^{(m)}_{i} (n_g)$
        		    \ENDIF
    		    \ELSE
    		        \STATE \%\% \textit{There is a pre-delay violation. Let $k_v$ denote the minimum index among violating packets}
    		        \STATE Allocate $\tau_{i} = \frac{(t_{k_v} - T_{\mathrm{pre},k_v}) - \sum_{i=1}^{k} \tau_i}{k_v - k}$ for $i=k+1,\dots,k_v$
    		        \STATE \textbf{break} and go to Line 23
    		    \ENDIF
    		\ENDWHILE \\
			
			\dotfill \\
            \STATE \%\% \textit{The Case with a Pre-Delay Violation} \\
			\WHILE{1}
			    \STATE $k \leftarrow k_v$
			    \STATE $w_r = w_{\mathrm{max}} - \sum_{i=1}^{k} w(\tau_i)$
			    \STATE $\tau_i = w^{-1}(\frac{w_r}{M-k})$, for $i=k+1,\dots,M$
			    \STATE Check among $[k+1,\dots,M]$ for pre-delay violations 
			    \IF{$\sum_{j=1}^{i} \tau_j \leq t_i + T_{\mathrm{pre},i}$ for all $i=k+1,\dots,M$}
			        \STATE \textbf{break} 
		        \ELSE
		            \STATE \%\% \textit{There is a pre-delay violation. Let $k_v$ denote the minimum index among violating packets}
    		        \STATE Allocate $\tau_{i} = \frac{(t_{k_v} - T_{\mathrm{pre},k_v}) - \sum_{i=1}^{k} \tau_i}{k_v - k}$ for $i=k+1,\dots,k_v$
			    \ENDIF
			\ENDWHILE \\
			\dotfill \\
			\textbf{return} $\boldsymbol{\tau}$
	\end{algorithmic}
\end{algorithm}

\par Algorithm \ref{alg:converse_twosided} operates in a similar fashion to Algorithm \ref{alg:converse_post}. However, unlike the post-delay only case, this need not translate to having non-increasing $\tau_i$ for all $i$. This follows from the fact that in a two-sided system, the optimal $\boldsymbol{\tau}$ would not have non-increasing $\tau_i$ for all $i$, in case there are pre-critical packets (see Lemma \ref{lem:optimality_conds_subgroup}). This is a direct result of energy optimality being a necessity for completion time optimality (Lemma \ref{lem:converse_needs_energy_optimal}), and Lemma \ref{lem:optimality_conds_subgroup}. Thus, instead of looking at the whole vector for non-increasing durations (as was done in Algorithm \ref{alg:converse_post}), Algorithm \ref{alg:converse_twosided} seeks to ensure non-increasing durations between the last subgroup of the ``left side" and the equalized durations on the ``right side", if the equal allocation on the right side complies with all pre-delay constraints of their respective packets.

\par In the existence of a pre-delay violation on the right side, Algorithm \ref{alg:converse_twosided} records the first (smallest index) of such violations (say $k_v$), and allocates equal durations for $\{k+1,\dots,k_v\}$ between $[\sum_{i=1}^{k} \tau_i, t_{k_v}+T_{\mathrm{pre},k_v}]$, where $k$ is the last packet of the previous subgroup. Note that for packets $\{1,\dots,k_v\}$, the time interval $[0,\sum_{i=1}^{k_v} \tau_i]$ cannot be made larger than $t_{k_v} + T_{\mathrm{pre},k_v}$, and this equal allocation, alongside the ``left side" (\emph{i.e.,} packets $\{1,\dots,k\}$), ensure energy optimality (follows identical from Appendix \ref{ap:proof_of_twosided}). For the remaining packets $\{k_v + 1,\dots,M\}$, we equalize their durations with the remaining energy budget and check their pre-delay violations once again. We repeat this process until there are no pre-delay violations in the last batch, which finishes the algorithm's operation and terminates it. Figure \ref{fig:converse_twosided} presents a toy example to illustrate the effects of a pre-critical packet. Note that Figure \ref{fig:converse_twosided} considers the identical problem as Figure \ref{fig:converse_postonly}, with the added pre-delay of $T_{\mathrm{pre},3} = 4.5$ for demonstrative purposes.

\begin{figure}[!t]
    \centering
    \includegraphics[width = 0.48\textwidth]{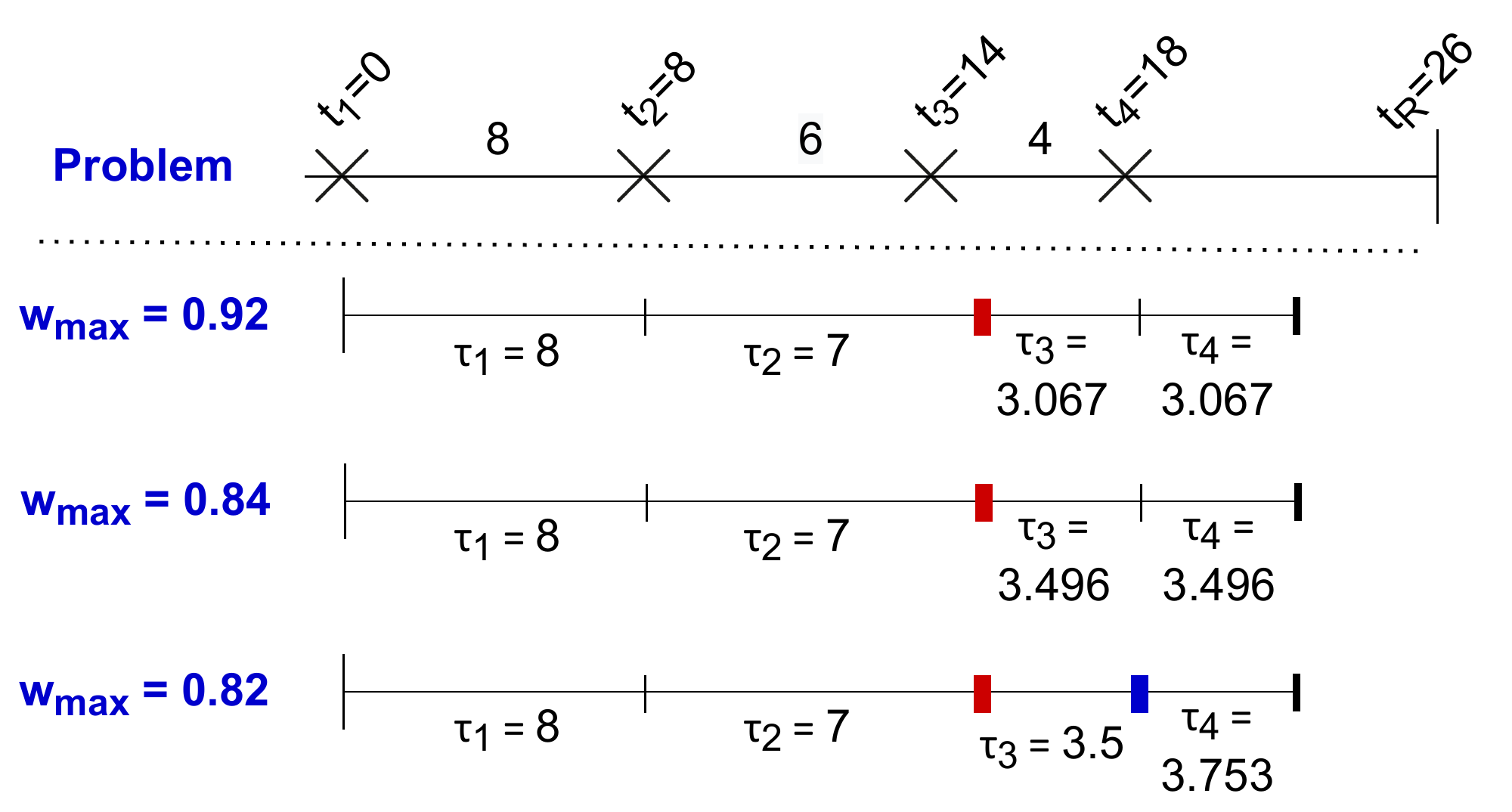}
    \caption{Demonstrative toy example for Algorithm \ref{alg:converse_twosided} under different $w_{\mathrm{max}}$ for $\boldsymbol{T}_{\mathrm{post}} = [\infty,11,\infty,\infty]$, $\boldsymbol{T}_{\mathrm{pre}} = [\infty,\infty,4.5,\infty]$, and $t_R=26$. The energy cost function $w(\tau)=\frac{1}{\tau}$. Blue/red marks on departure times represent pre-/post-criticality of the marked packet, respectively.}
    \label{fig:converse_twosided}
\end{figure}

\begin{theorem}
    \label{theo:converse_twosided_optimal}
    The output of Algorithm \ref{alg:converse_twosided} (say $\boldsymbol{\tau}$) yields the optimal transmission duration vector that minimizes $T_c$ under two-sided delay constraints. 
\end{theorem}
\begin{proof}
    The proof strategy is identical to the proof for the post-delay only case in Appendix \ref{ap:proof_of_comptime_post}. 
    \begin{itemize}
        \item For all cases except Case 2b-i, the proof follows from Lemma \ref{lem:converse_converse}. In particular, by design, Algorithm \ref{alg:converse_twosided} uses all available energy (\emph{i.e.,} $w_{\mathrm{max}}$). Showing the second requirement of the lemma, that is, the energy optimality of the resultant $\boldsymbol{\tau}$ within the interval $[0,\sum_{i=1}^{M} \tau_i]$, directly follows from the energy optimality proof in Appendix \ref{ap:proof_of_twosided}. Thus, the sufficient conditions in Lemma \ref{lem:converse_converse} are met, implying $\boldsymbol{\tau}$ minimizes the completion time.
        
        \item Similar to the post-delay only case, Case 2b-i corresponds to the case where i) $t_R - T_{\mathrm{post},M} > t_R - T_{\mathrm{post},i}$ for all $i \neq M$, ii) $t_R - T_{\mathrm{post},M} > t_M$, and iii) running Algorithm \ref{alg:general_twosided} between $[0,t_R - T_{\mathrm{post},M}]$ consumes $w(\boldsymbol{\tau}) \leq w_{\mathrm{max}}$ while achieving $\sum_{i=1}^{M} \tau_i = t_R - T_{\mathrm{post},M}$. Thus, the available energy can satisfy $T_c = t_R - T_{\mathrm{post},M}$, which is the smallest possible completion time. 
    \end{itemize}
\end{proof}

\noindent \mcg{\textbf{Computational Complexity:} For time minimization, Algorithms \ref{alg:converse_post} and \ref{alg:converse_twosided} rely on running the energy-minimizing Algorithm \ref{alg:general_twosided} first, which is $\mathcal{O}(M^3)$ in time and $\mathcal{O}(M)$ in space complexity at the worst case. The rest of the algorithms involve (at most $M$) comparisons of the elements of an $M$-element $\boldsymbol{\tau}$ vector, which incurs $\mathcal{O}(M^2)$ time and $\mathcal{O}(M)$ space complexity. Thus, the time minimization algorithms are also $\mathcal{O}(M^3)$ in time and $\mathcal{O}(M)$ in space.}

\section{\mcg{Numerical Results}}
\label{sec:numericalresults}

\subsection{Energy Optimization}
\label{subsec:energy_numerical}
\mcg{
\par Algorithm \ref{alg:general_twosided} generalizes the approaches in \cite{uysal2002_singledelay} and \cite{chen2008_predelay} to settings where two-sided delay constraints are present. In fact, in systems with no post-delay constraints ($T_{\mathrm{post},i} = \infty$), the system only has pre-delay constraints and Algorithm \ref{alg:general_twosided} is equivalent to the approach in \cite{chen2008_predelay}. If pre-delays are also absent ($T_{\mathrm{pre},i} = t_R-t_i$), the remaining system is identical to \cite{uysal2002_singledelay} (see Section \ref{sec:prelims}). Motivated by this, we present Figure \ref{fig:twosided_numerical}, which compares Algorithm \ref{alg:general_twosided} with its natural baselines: \cite{uysal2002_singledelay,chen2008_predelay}, and its own post-delay-only version (that disregards pre-delays).

\par Under two-sided delay constraints, schedulers that consider only pre-, post-, or no individual deadlines can potentially cause a delay violation on the sides they do not consider, leading to packet loss. To provide a fair comparison among schemes that might lose different amounts of packets, Figure \ref{fig:twosided_numerical} compares each scheme in terms of \emph{energy per successful packet transmission}, that is $\frac{w(\boldsymbol{\tau})}{\text{\# successful packets}}$. Given an arrival time $t_i$, the pre- and post-delays of each packet in the figure is such that its valid departure region is $[t_i + T, t_i + 2T]$ (\emph{i.e.,} $T_{\mathrm{pre},i} = 2T$ and $T_{\mathrm{post},i} = t_R - t_i - T$), with $T$ as the sweep parameter. Arrival times $t_i$ are chosen uniformly at random within $[0,t_R - 2T]$, and each data point on the figure is averaged over $10^4$ trials.


\begin{figure}
\centering
\begin{minipage}{.44\textwidth}
  \centering
  \includegraphics[width=.95\linewidth]{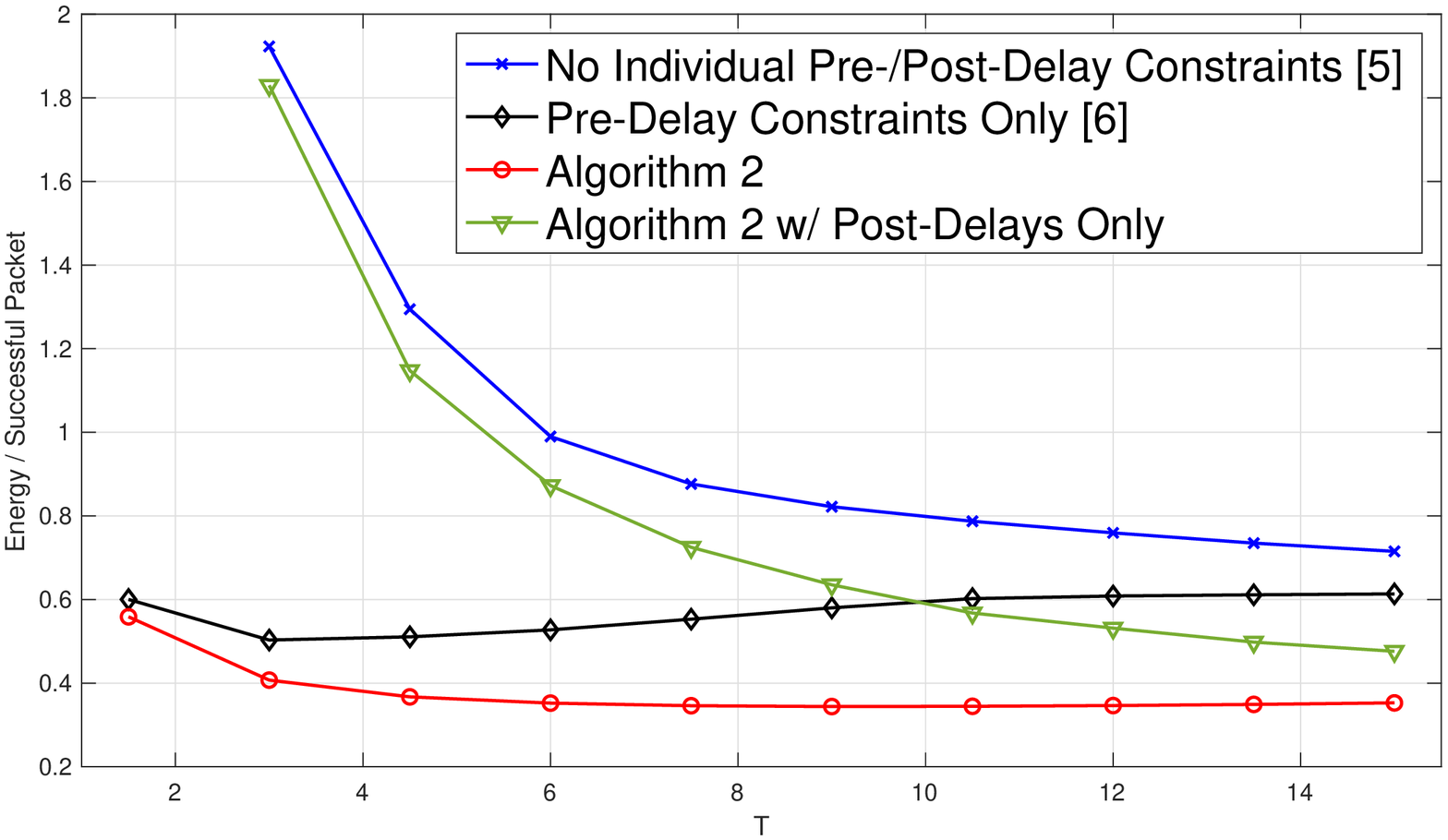}
  \captionof{figure}{ \mcg{Energy per successful packet vs. $T$ under two-sided delay constraints. $t_R = 100$, $M=30$, $w(\tau) = \frac{1}{\tau}$. } }
  \label{fig:twosided_numerical}
\end{minipage}%
\quad
\begin{minipage}{.44\textwidth}
  \centering
  \includegraphics[width=.95\linewidth]{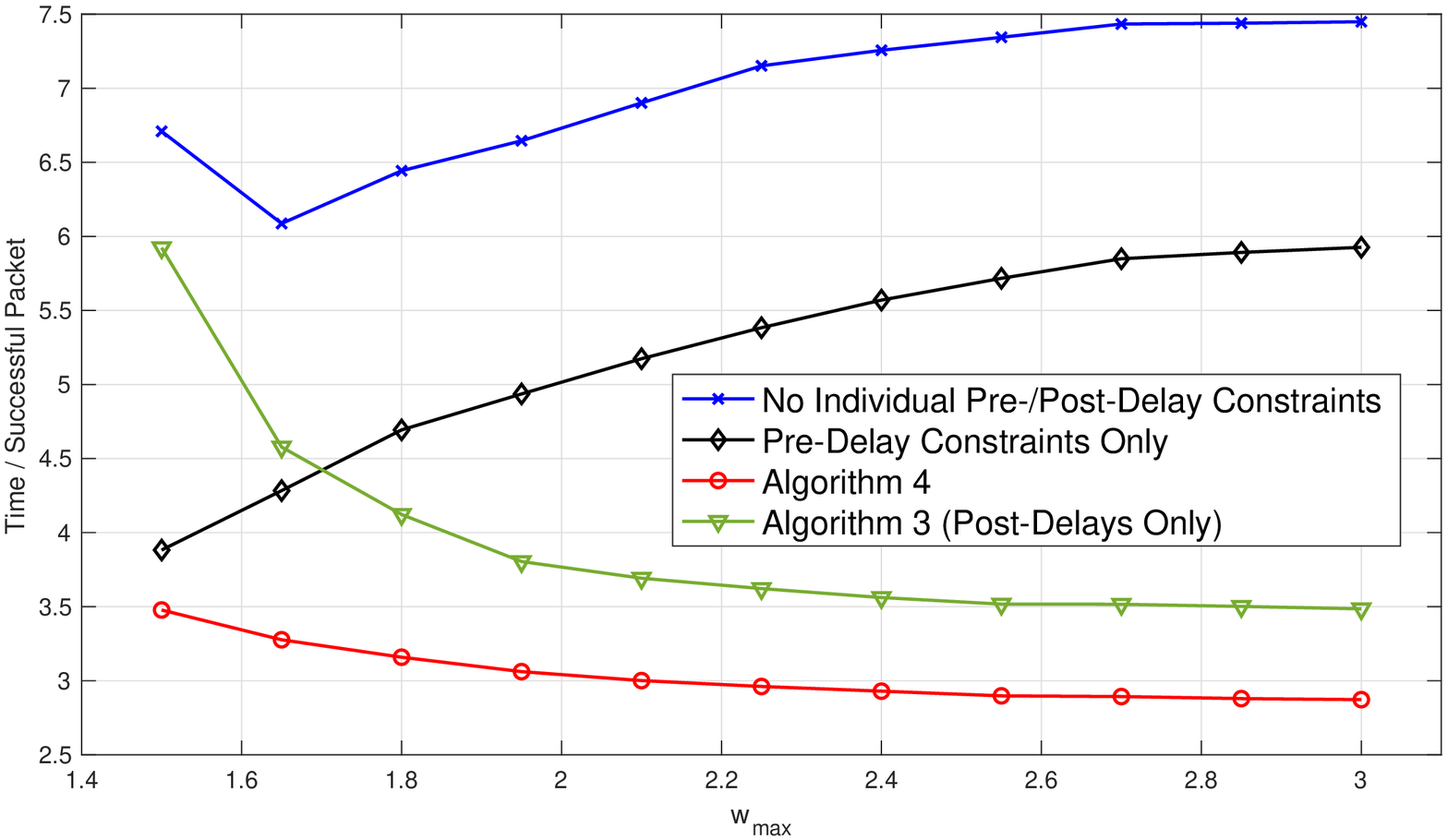}
  \captionof{figure}{\mcg{Time per successful packet vs. $w_{\mathrm{max}}$ under two-sided delay constraints. $t_R = 20$, $M=5$, $w(\tau) = \frac{1}{\tau}$, $T=3$.}}
  \label{fig:twosided_numerical_dual}
\end{minipage}
\end{figure}

\par In the small $T$ regime of Figure \ref{fig:twosided_numerical}, each packet's transmission duration is limited by a severely small pre-delay. The schemes that do not take this into consideration (no individual delays \cite{uysal2002_singledelay} and Algorithm \ref{alg:general_twosided} with post-delays only) fail in transmitting a packet successfully without a (pre-) delay violation. As $T$ increases, the pre-delays become less restrictive, and the valid departure regions of each packet are expanded, which together reduces the cost. However, as $T$ further increases, the post-delays become more restrictive, and the pre-delay only scheme (\cite{chen2008_predelay}) which does not consider post-delays in its allocation, drops more packets and its performance gets degraded. Overall, the results of Figure \ref{fig:twosided_numerical} show that in the evaluated system of interest wherein two-sided deadlines are present, Algorithm \ref{alg:converse_post} provides the lowest energy cost per successful packet in both smaller and larger $T$ regimes. }

\subsection{Completion Time Minimization}

\mcg{
\par Figure \ref{fig:twosided_numerical_dual} presents numerical results on the completion time minimizing scheme (Algorithm \ref{alg:converse_twosided}) with respect to the available energy budget $w_{\mathrm{max}}$, and compares it to its counterparts that only considers post-delays (Algorithm \ref{alg:converse_post}), only considers pre-delays, and no individual deadlines. Similar to the energy minimization case in Subsection \ref{subsec:energy_numerical}, each scheme's completion time is normalized on a per-successfully transmitted packet basis (\emph{i.e,} $\frac{T_c}{\text{\# successful packets}}$). Similar to Subsection \ref{subsec:energy_numerical}, each packet's valid region is $[t_i + T, t_i + 2T]$, and each data point is averaged over $10^4$ trials where $t_i$ are selected as random for each trial.


\par As expected, the completion time of Algorithm \ref{alg:converse_twosided} decreases with increasing $w_{\mathrm{max}}$, due to the decreasing nature of $w(\tau)$. On the other hand, the scheduler that only considers pre-delays is observed to perform worse as $w_{\mathrm{max}}$ increases. This is because it neglects post-delays, and therefore seeks to over-minimize the departure times with increasing $w_{\mathrm{max}}$, causing post-delay violations thus packet losses. Note that Algorithm \ref{alg:converse_post} does not face this issue, as the goal of the time-minimization problem is decreasing said departure times, which inherently addresses constraints due to pre-delays. The no-individual deadline scheme incurs bad performance due to limited energy in the small $w_{\mathrm{max}}$ regime, and as $w_{\mathrm{max}}$ increases, it runs into the same post-delay violation issue as the pre-delay-only algorithm. The results of Figure \ref{fig:twosided_numerical_dual} show that overall, Algorithm \ref{alg:converse_twosided} consistently yields the lowest completion time per successful packet. 

}


 
\section{Conclusions}
\label{sec:conclusions}

\par In this paper, the theoretical framework of packet scheduling under two-sided delay constraints has been addressed, wherein transmitting a packet too late (too stale) and too early (too fresh) are both undesired. To this end, firstly the energy-efficient packet scheduling problem under two-sided individual delay constraints has been introduced, which generalizes the conventional one-sided delay constrained scheduling. For the newly posed two-sided problem, an algorithm that yields the provably optimal transmission duration vector has been provided. Furthermore, the ``dual" problem has been formulated, which seeks to minimize transmission completion time subject to an energy budget, again under individual two-sided delay constraints. The algorithm that provably minimizes the transmission duration has been developed and presented.

\par The developed algorithms assume perfect and non-causal knowledge of arrival times and delay constraints, and are offline algorithms. Depending on the application, the necessary information of delay constraints or arrival times might be available or be accurately estimated, \mcg{in which case the algorithms developed herein can be directly employed or adapted}. However, for some applications, these critical information or perfect control of transmission rates/durations may not be feasible. \mcg{For these applications, the solution strategy of this paper can serve as a guideline for designing close-to-optimal, practical scheduler designs. From this perspective,} the primary future work of our study is devising more robust algorithms against aforementioned imperfections, while extending and adapting the developed framework to our motivating applications. Other future works include considering non-FIFO ordering and non-convex cost functions, and expanding the two-sided delay constrained framework to multi-hop systems, \mcg{as well as systems with multiple transmitter/receiver nodes}.


\appendices

\section{Proof of Theorem \ref{theo:optimality_twosided}}
\label{ap:proof_of_twosided}

\par The proof is similar in spirit as the pre-delay-only (\cite{chen2008_predelay}) case, in that the proof for the two-sided case also follows from the concept of majorization. Compared to the post-delay-only case, as the subsequent subgroups need not be of non-increasing order, they need to be taken into account for the two-sided case (similar to the pre-delay-only case). Furthermore, compared to the pre-delay-only case, herein, the deadline constraints introduced by both pre- and post-delays need to be considered, which will be reflected in Claims 1-3.

\begin{definition}[\emph{Majorization}]
    \emph{An $M$-element, non-increasing sequence $\boldsymbol{x}_1$ is said to be majorized by another $M$-element, non-increasing sequence $\boldsymbol{x}_2$ if 
    \begin{itemize}
        \item $\sum_{i=1}^{M} x_{1,i} = \sum_{i=1}^{M} x_{2,i}$, and 
        \item $\sum_{i=1}^{m} x_{1,i} \leq \sum_{i=1}^{m} x_{2,i}$, for all $m=1,\dots,M-1$.
    \end{itemize}}
\end{definition}
Note that $\boldsymbol{x}_1$ being majorized by $\boldsymbol{x}_2$ implies that summing up to the same value, $\boldsymbol{x}_1$ is more ``balanced" than $\boldsymbol{x}_2$ \cite{majorization_ref2}. 
    
    \par Recalling from Section \ref{sec:system_model}, the cost $w(\boldsymbol{\tau}) = \sum_{i=1}^{M} w(\tau_i)$ does not depend on the order of the $M$-element sequence $\boldsymbol{\tau}$ and is strictly convex in each $\tau_i$, which makes it Schur-convex (order-preserving). Therefore, if a transmission vector $\boldsymbol{\tau}_1$ is majorized by another transmission vector $\boldsymbol{\tau}_2$, it is guaranteed to have
    $\sum\nolimits_{i=1}^{M} w\left(\tau_{1_i}\right) \leq \sum\nolimits_{i=1}^{M} w\left(\tau_{2_i}\right).$ Motivated by this, we will show herein that the output of Algorithm \ref{alg:general_twosided} (say $\boldsymbol{\tau}$ throughout the proof) gets majorized by any other feasible sequence $\boldsymbol{\tau}'$, which guarantees that $\boldsymbol{\tau}$ incurs the lowest cost among all feasible transmission durations. In particular, we will show 
    the statement by showing the contrapositive ``if some $\boldsymbol{\tau}'$ does not majorize $\boldsymbol{\tau}$, then it is not feasible". 
    
    \par We first start by ordering $\boldsymbol{\tau}'$ and $\boldsymbol{\tau}$ in a non-increasing manner, and define $\boldsymbol{\tau}'^{\downarrow}$ and $\boldsymbol{\tau}^\downarrow$ as the non-increasingly ordered versions of $\boldsymbol{\tau}'$ and $\boldsymbol{\tau}$, respectively. Let $j$ be the smallest index such that $\sum_{i=1}^{j} \tau^\downarrow_{i} > \sum_{i=1}^{j} \tau'^{\downarrow}_{i}$. Then, we know $\tau^\downarrow_{j} > \tau'^{\downarrow}_{j}$. Let $j_{c}$ be the largest index such that $\tau_j = \tau_{j_c}$. Essentially, $\tau_{j_c}$ ``completes" the last incomplete group/subgroup in the batch, yielding an \emph{integer multiple of completed subgroups}. Since $\boldsymbol{\tau}'^{\downarrow}$ is ordered in a non-increasing manner, we have $\tau'^{\downarrow}_{j_c} \leq \dots \leq \tau'^{\downarrow}_j < \tau^{\downarrow}_j = \dots = \tau^{\downarrow}_{j_c},$
    which implies 
    $$\sum\nolimits_{i=1}^{j_c} \tau^\downarrow_{i} > \sum\nolimits_{i=1}^{j_c} \tau'^{\downarrow}_{i}.$$ 
    
    \par Let the collection of packets $\{1,\dots,j_c\}$ of the ordered $\boldsymbol{\tau}^{\downarrow}$ correspond to packets $\{k_1,\dots,k_{j_c}\}$ of the un-ordered, original $\boldsymbol{\tau}^{\downarrow}$. Now let us consider the whole set of packets $\{1,\dots,k_{j_c}\}$ of the original $\boldsymbol{\tau}$. We have two cases: \\
    \textbf{Case 1:} The first case is where no pre-critical subgroups are among the subgroups in $\{1,\dots,k_{j_c}\}$ of $\boldsymbol{\tau}$. For this case, as a corollary of Lemma \ref{lem:optimality_conds_subgroup}, the set $\{1,\dots,j_c\}$ of the ordered $\boldsymbol{\tau}^{\downarrow}$ corresponds to the set $\{1,\dots,j_c\}$ of the un-ordered, original $\boldsymbol{\tau}$, with $k_{j_c} = j_c$.\footnote{Note that if a subgroup exists in $\{1,\dots,k_{j_c}\}$ that is not in $\{1,\dots,j_c\}$ of the ordered $\boldsymbol{\tau}^{\downarrow}$, then the duration of the subsequent subgroup (the first one that is in $\{1,\dots,j_c\}$ of $\boldsymbol{\tau}^{\downarrow}$) would need to have a larger transmission duration than the one that is not in $\{1,\dots,j_c\}$ of $\boldsymbol{\tau}^{\downarrow}$. This would require an \emph{increase} in transmission duration from one subgroup to the subsequent, which can only be achievable by a pre-critical subgroup, yielding a contradiction.} Since there are no pre-critical subgroups in the batch, the last subgroup is either a regular or a post-critical one. \\
    \par \textit{Case 1a:} If index ${j_c}$ corresponds to the last packet of a regular subgroup, then by definition we have $\sum_{i=1}^{j_c} \tau_{i} = \sum_{i=1}^{j_c} d_i$, which yields $\sum_{i=1}^{j_c} \tau'_{i} < \sum_{i=1}^{j_c} d_i$. This is in violation of the non-idling constraints in \eqref{eq:optimization_nonidling1}-\eqref{eq:optimization_nonidling2}, and implies $\boldsymbol{\tau}'$ is infeasible. \\
    \par \textit{Case 1b:} If index ${j_c}$ corresponds to the last packet of a post-critical subgroup, then by definition we have $\sum_{i=1}^{j_c} \tau_{i} = t_R - T_{\mathrm{post},j_c}$, which yields $\sum_{i=1}^{j_c} \tau'_{i} < t_R - T_{\mathrm{post},j_c}$. This is in violation of the post-transmission delay constraint in \eqref{eq:optimization_postdelay}, and implies $\boldsymbol{\tau}'$ is infeasible.
    \\
    \textbf{Case 2:} The second case is where there are pre-critical subgroups among the subgroups in $\{1,\dots,k_{j_c}\}$ of $\boldsymbol{\tau}$. We first make the following claim, with accompanying proof in Footnote \ref{foot:twosidedproof_case2}.
    \\
    \emph{Claim 1:} The last the last subgroup in $\{1,\dots,k_{j_c}\}$ of $\boldsymbol{\tau}$ cannot be a pre-critical subgroup.\footnote{\label{foot:twosidedproof_case2} Suppose the last subgroup was pre-critical. Then the subsequent subgroup would have had to have a larger transmission duration than it, following Item 3 of Lemma \ref{lem:optimality_conds_subgroup}. Hence, the subsequent subgroup would also have been included in $\{1,\dots,j_c\}$ of the ordered $\boldsymbol{\tau}^{\downarrow}$, hence the initial subgroup would not have been the last one, yielding a contradiction.} \\
    We partition $\{1,\dots,k_{j_c}\}$ of $\boldsymbol{\tau}$ into two subsets, noting
    \begin{equation}
        \label{eq:partition}
        \sum\nolimits_{i=1}^{k_{j_c}} \tau_i = \sum_{i \in \{k_1,\dots,k_{j_c}\}}^{}\tau_i + \sum_{i \notin \{k_1,\dots,k_{j_c}\}}^{} \tau_i
    \end{equation}
    We first note that by definition of the set $\{k_1,\dots,k_{j_c}\}$ of $\boldsymbol{\tau}$, every element in the first subset is strictly greater than every element in the second subset. For the first subset, we have 
    \begin{equation}
        \begin{split}
            \sum\nolimits_{i \in \{1,\dots,k_{j_c}\}}^{}\tau_i = \sum\nolimits_{i=1}^{j_c} \tau^\downarrow_{i} > \sum\nolimits_{i=1}^{j_c} \tau^{\downarrow'}_{i} \geq \sum\nolimits_{i=1}^{j_c} \tau'_{i}.\\
        \end{split}
    \end{equation}
    For the second subset, we first observe that it consists of collections that contain one or more, integer multiple of consecutive subgroups. Furthermore, we make the following claims, with proofs in Footnotes \ref{foot:twosidedproof_case2_claim1}-\ref{foot:twosidedproof_case2_claim2}: \\
    \textit{Claim 2:} Each collection ends with a pre-critical subgroup.\footnote{\label{foot:twosidedproof_case2_claim1} Assume otherwise. Then, by Items 1 and 2 of Lemma \ref{lem:optimality_conds_subgroup}, the subsequent subgroup of a post-critical or regular subgroup would have a smaller transmission duration, hence would also be included in the second subset, yielding a contradiction.} \\
    \textit{Claim 3:} A subgroup that is the immediate prior to a collection cannot be a pre-critical subgroup.\footnote{\label{foot:twosidedproof_case2_claim2} Assume the prior subgroup is pre-critical. Then, by Item 3 of Lemma \ref{lem:optimality_conds_subgroup}, the prior subgroup would have a smaller transmission duration, hence would also be included in the second subset hence the collection, yielding a contradiction.}
    \\ Combining Claims 2-3, each collection within the second subset starts at either the arrival of the first packet in the collection (if prior subgroup is regular), or the post-transmission constraint of the last packet of the prior subgroup (if prior subgroup is post-critical). Furthermore, each collection ends with a pre-critical subgroup. Thus, the sums of the transmission durations in each collection are at their maximum possible and cannot be increased further. Otherwise, either a pre-/post-transmission constraint would be violated, or a packet's transmission would start before its arrival, which is clearly a violation. Therefore, we know that for any $\boldsymbol{\tau}'$, we have to have
    \begin{equation}
        \label{eq:notin}
        \begin{split}
            \sum\nolimits_{i \notin \{1,\dots,k_{j_c}\}}^{}\tau_i \geq \sum\nolimits_{i \notin \{1,\dots,k_{j_c}\}}^{} \tau'_i.\\
        \end{split}
    \end{equation}
    Combining equations \eqref{eq:partition}-\eqref{eq:notin}, we have $\sum_{i=1}^{k_{j_c}} \tau_i
    > \sum_{i=1}^{k_{j_c}} \tau'_i$. From Claim 1, the last subgroup in $\{1,\dots,k_{j_c}\}$ of $\boldsymbol{\tau}$ is either regular or post-critical. With a similar argument to Case 1, if the last subgroup is regular, then $\sum_{i=1}^{k_{j_c}} \tau'_i$ violates the non-idling constraint; and if the last subgroup is post-critical, then $\sum_{i=1}^{k_{j_c}} \tau'_i$ violates the post-transmission constraint of the $k_{j_c}$th packet, which makes $\boldsymbol{\tau}'$ infeasible in both scenarios. 
    
    \par Overall, whenever $\boldsymbol{\tau}'$ does not majorize $\boldsymbol{\tau}$, all possible scenarios yield $\boldsymbol{\tau}'$ infeasible. Thus, $\boldsymbol{\tau}$ gets majorized by all other feasible sequences, which concludes the proof. 
    
\section{Proof of Theorem \ref{theo:converse_post_optimal}}
\label{ap:proof_of_comptime_post}   
    
    We will proceed in a case-by-case manner.
    \\
    \noindent \textbf{Case 2b-i:} This is the case where the post-delay-induced departure lower bound of the $M$th packet occurs last, $t_M < t_R - T_{\mathrm{post},M}$, and the available $w_{\mathrm{max}}$ can achieve $\sum_{i=1}^{M} \tau_i = t_R - T_{\mathrm{post},M}$. Since $T_c \geq t_R - T_{\mathrm{post},M}$ for a FIFO system, critically achieving it achieves the smallest possible completion time, which makes $\boldsymbol{\tau}$ optimal.
    \\
    \noindent \textbf{Cases 1, 2a, and 2b-ii:} These are the cases where $\sum_{i=1}^{M} \tau_i > t_R - T_{\mathrm{post},M}$. Leveraging Lemma \ref{lem:converse_converse}, it suffices to show that i) $w(\boldsymbol{\tau}) = w_{\mathrm{max}}$, and ii) $\boldsymbol{\tau}$ minimizes the energy cost within $[0,\sum_{i=1}^{M} \tau_i]$.
    \begin{enumerate}[\listparindent=1.5em]
        \item Showing $w(\boldsymbol{\tau}) = w_{\mathrm{max}}$ is trivial for all Cases 1, 2a, and 2b-ii, as the algorithm is guaranteed to allocate all available $w_{\mathrm{max}}$ in all cases.
        
        \item \textbf{Energy Optimality of Case 1:} As we are proving energy optimality, the proof will again follow by the majorization argument, similar to the proof of Algorithm \ref{alg:general_twosided}.
        
        \par Let $\boldsymbol{\tau}' \neq \boldsymbol{\tau}$ with $\sum_{i=1}^{M} \tau'_i = \sum_{i=1}^{M} \tau_i.$ Note that by design of Algorithm \ref{alg:converse_post} and Corollary \ref{cor:converse_post_decreasing}, both $\boldsymbol{\tau}'$ and $\boldsymbol{\tau}$ are of non-increasing order. 
        Suppose $\boldsymbol{\tau}'$ does not majorize $\boldsymbol{\tau}$, and let $j$ be the smallest index such that $\sum_{i=1}^{j} \tau_{i} > \sum_{i=1}^{j} \tau'_{i}$. Similar to Appendix \ref{ap:proof_of_twosided}, we know $\tau_{j} > \tau'_{j}$. Again, denote $j_{c}$ as the largest index such that $\tau_j = \tau_{j_c}$, thereby yielding $\tau'_{j_c} \leq \dots \leq \tau'_j < \tau_j = \dots = \tau_{j_c},$ and hence
        \begin{equation}
            \label{eq:inequality_majorizationproof}
            \sum\nolimits_{i=1}^{j_c} \tau_{i} > \sum\nolimits_{i=1}^{j_c} \tau'_{i}. 
        \end{equation}
        We have two cases:
        \\
        \textbf{The case with $j_c < M$:} Say Algorithm \ref{alg:converse_post} stops at packet index $k$. By design, the algorithm partitions the set of packets, as well as the time axis $[0,\sum_{i=1}^{M} \tau_i]$ into two parts (we will call them left and right sides): packets $\{1,\dots,k\}$ and $\{k+1,\dots,M\}$, with corresponding intervals $[0,\sum_{i=1}^{k} \tau_i]$ and $[\sum_{i=1}^{k} \tau_i,\sum_{i=1}^{M} \tau_i]$, respectively. The case with $j_c < M$ implies that $j_c$ is on the left side. Then, by the design of Algorithm \ref{alg:converse_post}, we have either $\sum\nolimits_{i=1}^{j_c} \tau_{i} < \sum_{i=1}^{j_c} d_i$ or $\sum\nolimits_{i=1}^{j_c} \tau_{i} < t_R - T_{\mathrm{post},j_c}$. Under both these cases, having \eqref{eq:inequality_majorizationproof} implies $\boldsymbol{\tau}'$ violates either a non-idling or post-delay constraint at packet $j_c$, which makes it infeasible.
        \\
        \textbf{The case with $j_c = M$:} For this case, we have $\sum_{i=1}^{M} \tau'_i < \sum_{i=1}^{M} \tau_i$, which violates the proof's initial supposition of having $\sum_{i=1}^{M} \tau'_i = \sum_{i=1}^{M} \tau_i$, yielding a contradiction.
        \\
        Thus, for any $j_c$, if $\boldsymbol{\tau}'$ does not majorize $\boldsymbol{\tau}$, it is either infeasible or causes a contradiction.
    \end{enumerate}
    \par Since we showed both $w(\boldsymbol{\tau}) = w_{\mathrm{max}}$ and the energy optimality of $\boldsymbol{\tau}$ between $[0,\sum_{i=1}^{M} \tau_i]$, from Lemma \ref{lem:converse_converse}, $\boldsymbol{\tau}$ is the completion time minimizing vector with the optimal completion time $T^*_c = \sum_{i=1}^{M} \tau_i$. The proofs for Cases 2a and 2b-ii follow identically.

\bibliographystyle{IEEEtran}
\bibliography{refs_GLOBECOM}


\end{document}